\documentclass[11pt]{article}

\usepackage[margin=0.9in]{geometry}

\usepackage{titlesec,fancyhdr,etoolbox}
\usepackage{booktabs}
\usepackage{multirow}
\usepackage{threeparttable}
\usepackage{adjustbox}
\usepackage{amsmath,amssymb}
\usepackage{iftex}
\usepackage{amsfonts}
\usepackage{multirow}
\usepackage{amsthm, bm}
\usepackage{xcolor}
\usepackage{graphics} 
\usepackage{algorithm}
\usepackage{algpseudocode}
\usepackage{booktabs}
\usepackage{lineno}
\allowdisplaybreaks

\usepackage{natbib}

\usepackage[
    colorlinks=true,
    citecolor=blue,
    linkcolor=blue,
    urlcolor=blue
]{hyperref}

\newtheorem{theorem}{Theorem}[]
\newtheorem{definition}[theorem]{Definition}

\newtheorem{lemma}[theorem]{Lemma}

\newtheorem{remark}{Remark}

\newcommand{\R}{\mathbb{R}}
\newcommand{\kcl}{{(k)}}
\newcommand{\kclp}{{(l)}}
\newcommand{\kone}{{(1)}}
\newcommand{\ktwo}{{(2)}}

\newcommand{\kkm}{{(K-1)}}
\newcommand{\kk}{{(K)}}
\newcommand{\Kkk}{K_{k,l}}
\newcommand{\none}{{n_1}}
\newcommand{\ntwo}{{n_2}}
\newcommand{\nK}{{n_K}}
\newcommand{\nk}{{n_k}}
\newcommand{\nkk}{{n_{l}}}
\newcommand{\lamk}{\lambda_{k,l}}
\newcommand{\lamkk}{\lambda_{l,k}}
\newcommand{\Exp}{{E}}
\newcommand{\Prob}{P} 

\newcommand{\snk}{\sum_{i = 1}^{n_k}}
\newcommand{\snkk}{\sum_{i = 1}^{n_{l}}}
\newcommand{\sd}{\sum_{j = 1}^{p}}
\newcommand{\sK}{\sum_{k = 1}^{K}}

\newcommand{\G}{{{G}_p}}
\newcommand{\CBt}{{Q_{\alpha,2}}}
\newcommand{\CBp}{{Q_{\alpha,pair}}}
\newcommand{\CBk}{{Q_{\alpha,K}}}
\newcommand{\kt}{k=1,2}
\newcommand{\kK}{k=1,\dots,K}
\newcommand{\kKK}{1\leq k<l \leq K}
\newcommand{\sm}{\sum_{m = 1}^{K(K+1)p/2}}
\newcommand{\st}{\sum_{t = 1}^{N}}

\title{\bfseries Testing Microbiome Community Differences in High Dimensions: A Bootstrap Approach for Compositional Data}

\author{
Monika Bhattacharjee \thanks{Department of Mathematics, Indian Institute of Technology Bombay}
\and
Nilanjan Chakraborty \thanks{Department of Mathematics and Statistics, Missouri University of Science and Technology}
\and
Sayan Das \thanks{Department of Statistics and Data Science, Washington University in St. Louis}
\and
Sounak Chakraborty \thanks{Department of Statistics and Data Science, University of Missouri}
\and
Lei Liu \thanks{Division of Biostatistics, Washington University in St. Louis}
\and
Yiming Shi $^\P$
\and
Kristine M. Wylie \thanks{Department of Pediatrics, Washington University in St. Louis}
\and
Todd N. Wylie$^{\;\;|\!|}$
\and
Molly J. Stout\thanks{Department of Obstetrics and Gynecology, University of Michigan
}
}
\date{}
\begin{document}
\maketitle

\begin{abstract}
Understanding differences in microbial community structure is critical for uncovering risk factors and mechanisms underlying diseases such as colorectal cancer and preterm birth. Microbiome data present unique statistical challenges because they are compositional in nature, violating assumptions of many classical inference procedures. We propose an empirical bootstrap framework that enables robust hypothesis testing for equality of microbial community means across groups, including two-sample, paired, and multi-sample settings. The method accounts for the simplex structure of microbiome data and provides valid inference even in high-dimensional regimes. Through applications to two large-scale studies, fecal microbiota in colorectal adenoma and cancer patients, and vaginal microbiota in pregnancy with preterm birth outcomes—we demonstrate that our approach identifies clinically meaningful differences that conventional methods fail to detect, such as age-related differences in adenoma prevalence and race-associated disparities in vaginal microbiome composition. These results highlight the potential of resampling-based inference for advancing microbiome research, improving reproducibility, and uncovering clinically relevant microbial signatures. 
\end{abstract} 

\vspace{0.5em}

\noindent
\textbf{Keywords:}
Conformal prediction;
optimal sample allocation;
asymptotic optimality;
distribution-free inference.

\vspace{0.3em}

\section{Introduction}
\label{s:intro}
The human microbiome plays a critical role in health and disease, influencing conditions ranging from colorectal cancer to preterm birth. Advances in sequencing technology have led to large-scale studies, such as the Human Microbiome Project (HMP), generating high-dimensional datasets that measure the relative abundance of thousands of microbial taxa across comparatively few individuals. These datasets are compositional in nature, with all microbial proportions constrained to sum to one. This structural feature makes many conventional statistical methods (e.g., the two-sample $t$-test, Wilcoxon signed-rank test, or classical goodness-of-fit tests) inappropriate and potentially misleading. The challenges of analyzing compositional data were first formalized by \citet{aitchison1982statistical}, who introduced log-ratio transformations to adapt standard inferential methods, and later expanded in his monograph \citep[Chapter~7]{aitchison2003compositional}. Building on these foundations, a variety of statistical approaches have been proposed for compositional data, including multivariate analysis of variance methods \citep{anderson2014permanova}, microbiome-specific frameworks \citep{li2015microbiome}, and bias-corrected differential abundance testing \citep{banerjee2018keystone, lin2020ancombc}. A comprehensive review of these techniques is provided in \citet{greenacre2021compositional}.

Despite these advances, much of the literature treats microbial taxa individually, leading to reduced power when the primary interest lies in detecting global differences across complex microbial communities. Multivariate strategies that consider the full microbial composition are more appropriate in these settings. For example, \citet{cao2017compositional} proposed a max-norm based test for assessing differences in mean compositions, demonstrating asymptotic consistency under restrictive conditions such as equal covariance structures. Later, \citet{li2022maximum} extended these ideas to settings with unequal covariance matrices, but both approaches rely on large sample sizes for their asymptotic approximations and remain limited to two-sample comparisons.

More recently, bootstrap-based procedures have been investigated for high-dimensional mean testing \citep{xue2020distribution, chakraborty2023multiplier, wu2024multisample}. However, most of these methods are not directly applicable to microbiome data. For example, multiplier bootstrap techniques \citep{xue2020distribution, chakraborty2023multiplier} cannot incorporate the simplex structure of compositional data, while parametric bootstrap approaches \citep{wu2024multisample} require covariance estimation that is unstable when matrices are nearly singular in high dimensions. These limitations call for new microbiome-specific approaches.

In this article, we propose EBC (Empirical Bootstrap for Compositional Data), an empirical bootstrap-based testing procedure for high-dimensional compositional data. The proposed max-norm test applies to two-sample, paired-sample, and multi-sample settings, while accounting for the compositional nature of microbiome data through CLR transformation and empirical bootstrap. We establish its theoretical validity under mild regularity conditions and demonstrate through extensive simulations that EBC maintains nominal type I error while achieving competitive or improved power relative to existing methods. Finally, we illustrate its practical utility using two clinically relevant microbiome studies involving colorectal cancer and preterm birth. These case studies demonstrate the practical utility of the proposed method for reliable inference in high-dimensional microbiome studies. 

The remainder of this paper is organized as follows. Section~\ref{why} discusses the clinical motivation for comparing microbial communities in colorectal cancer and preterm birth studies. Section~\ref{testtwo} introduces the proposed EBC framework for two-sample, paired-sample, and multi-sample testing, while Section~\ref{algo} describes the empirical bootstrap algorithm used to compute the data-driven critical values. Section~\ref{assum} presents the regularity assumptions, and Section~\ref{thr} establishes the theoretical guarantees of the proposed procedure. Section~\ref{simu} reports the results of simulation studies comparing EBC with existing methods, and Section~\ref{realdata} illustrates its performance using two real-world microbiome datasets. Finally, Section~\ref{Discussion} concludes with a discussion. Theoretical proofs and additional simulation details are provided in Section~\ref{App}.


\vspace{-0.08in}
\section{Why This Matters Clinically: Microbiome Differences in Colorectal Cancer and Preterm Birth} \label{why}


Colorectal cancer (CRC) remains a major global health challenge, with more than 1.9 million new cases and approximately 904{,}000 deaths estimated in 2022, making CRC one of the most commonly diagnosed cancers worldwide and a leading cause of cancer mortality \citep{Bray2024CA}. In the United States, prognosis depends strongly on stage at diagnosis: contemporary SEER estimates indicate 5-year relative survival of about 91\% for localized disease, 74\% for regional, and 16\% for distant-stage CRC \citep{SEERStatFactsCRC}. Preterm birth (PTB) likewise poses a substantial global public health burden: an estimated 13.4 million babies (1 in 10 live births) were born too soon in 2020, and PTB remains the leading cause of under-5 mortality worldwide \citep{BornTooSoon2023,WHO_PTB_2023}. In the United States, the provisional national PTB rate was 10.41\% in 2023, with persistent disparities by race and ethnicity \citep{CDC_VSRR_2023}.

Cancer care imposes substantial costs on health systems and patients. In the U.S., total cancer care costs were projected at \$208.9\,billion in 2020, with CRC among the highest contributors to medical services spending \citep{NCI_CTPR_EconBurden}. CDC analyses similarly note that CRC accounts for about 11.6\% of all cancer treatment costs, with estimated 2020 spending of \$23.7\,billion on medical services and \$0.6\,billion on prescription drugs \citep{CDC_CRC_costs}. For PTB, the lifetime societal cost for the 2016 U.S.\ birth cohort was estimated at \$25.2\,billion (present value), averaging about \$64{,}800 per preterm birth, driven largely by neonatal/infant medical care and longer-term developmental services \citep{Waitzman2021}. These cost gradients worsen when prevention, risk stratification, and timely interventions are suboptimal.

For CRC, tumor-associated dysbiosis provides both mechanistic insight and potential biomarkers. \textit{Fusobacterium nucleatum} is repeatedly enriched in CRC tissue and higher intratumoral loads are linked to recurrence, metastasis, and poorer outcomes; recent work identified a tumor-adapted \textit{F.\ nucleatum} clade with disease-relevant features \citep{ZepedaRivera2024}. In parallel, colibactin-producing (\emph{pks}$^{+}$) \emph{Escherichia coli} induces a distinctive genotoxic mutational signature that has been detected in human intestinal organoids and in patient tissues, supporting a plausible causal link between microbial genotoxins and colorectal carcinogenesis \citep{PleguezuelosManzano2020,Chen2023}. For PTB, vaginal microbiome community state types (CSTs) with low \emph{Lactobacillus}, particularly depletion of \emph{L.\ crispatus}-and greater anaerobe diversity (e.g., \emph{Gardnerella}) are consistently associated with elevated spontaneous PTB risk; meta-analytic and multi-omic studies show stronger predictability of early PTB from community composition and functional potential \citep{Huang2023, Gudnadottir2022, Liao2023}.

In CRC, insufficient screening uptake (despite USPSTF grade~A/B recommendations for average-risk adults 45--75 years) shifts diagnoses toward advanced stages with markedly worse survival and higher downstream costs; improving coverage and using validated biomarkers could both reduce mortality and lower spending \citep{USPSTF2021, NCI_CTPR_EconBurden, CDC_CRC_costs}. In PTB, failure to identify and manage infection- or microbiome-linked risk early in pregnancy can increase neonatal intensive care utilization and long-term disability, amplifying the multibillion-dollar societal burden \citep{Waitzman2021, BornTooSoon2023}. Because these questions hinge on detecting \emph{multivariate} differences in microbial compositions between clinical groups, methods that respect compositional constraints and retain power in high dimensions are directly responsive to medical decision-making in CRC and PTB.

\section{Empirical Bootstrap Tests for High-Dimensional Compositional Data } \label{testtwo}


This section introduces a practical testing framework for comparing mean compositions in high-dimensional microbiome (and other compositional) data while respecting the zero-sum constraint induced by CLR transformation. In Subsection~\ref{twotest}, we introduce the preliminaries and formulate the two-sample problem, the most common clinical comparison (e.g., cases vs.\ controls). We then extend it to paired designs in Subsection~\ref{paired} to accommodate within-subject or before and after comparisons frequently encountered in longitudinal or matched clinical settings. Subsection~\ref{multitest} further generalizes the approach to multi-arm studies, enabling inference across more than two clinical groups. Finally, Subsection~\ref{algo} details the data-driven procedure (empirical bootstrap) used to compute the test statistic and its critical value, providing readers with clear implementation steps for applied analyses.

\subsection{Two-sample Test for Independent Populations} \label{twotest}
\subsubsection{Some Preliminaries of Compositional Data } 
We define $\mathcal{S}^{p-1}$ as the $(p-1)$-dimensional simplex where $\mathcal{S}^{p-1} = \{(x_1,\dots,x_p) \in \R^p: x_j>0,\; j=1,\dots,p \text{ and } \sd x_j = 1\}$. Let $X^\kone = (X_1^\kone, \dots, X_\none^\kone)^T$ and $X^\ktwo = (X_1^\ktwo,\dots, X_\ntwo^\ktwo)^T$ be the observed data matrices, where $X^\kone$ and $X^\ktwo$ are collections of independent random vectors in $\mathcal{S}^{p-1}$. Further, we denote $W^\kone = (W_1^\kone, \dots, W_\none^\kone)^T$ and $W^\ktwo = (W_1^\ktwo, \dots, W_\ntwo^\ktwo)^T$ as the unobserved bases that generates the observed compositional data $X^\kone$ and $X^\ktwo$ respectively via the following normalization,
\begin{align*}
    X_{ij}^\kcl = W_{ij}^\kcl\Big{/}\sd W_{ij}^\kcl,\; i=1,\dots,\nk,\;j=1\dots,p,\;\kt,
\end{align*}
where $X_{ij}^\kcl$ and $W_{ij}^\kcl$ are the $j$-th elements of $X_i^\kcl$ and $W_i^\kcl$ respectively.

Let $V^\kcl = (V_1^\kcl, \dots, V_\nk^\kcl)^T$ be the log-basis vectors, where $V_{ij}^\kcl = \log W_{ij}^\kcl,\;\kt.$ For $\kt$, suppose that $V_1^\kcl,\dots, V_\nk^\kcl$ are two independent samples, each from a distribution with mean $\mu^\kcl = (\mu_1^\kcl, \dots, \mu_p^\kcl)^T$ and covariance matrix $\Sigma_i^\kcl = (\sigma_{i, jj'}^\kcl)_{jj'=1}^p,\; i = 1,\dots,n_k.$ Due to the simplicial structure of the observed data, the null hypothesis
\begin{align} \label{hypo}
H_{01}: \mu^\kone = \mu^\ktwo \text{ vs. } H_{a1}: \mu^\kone \neq \mu^\ktwo
\end{align}
is not testable (cf. \citet{cao2017compositional}). Instead, the above hypothesis is testable up to an additive constant. Towards that, the testable hypothesis can be written as
\begin{equation} \label{H01}
\begin{split}
    &H_{02}: \mu^\kone = \mu^\ktwo + c1_p, \text{ for some $c\in\R,$ vs. } H_{a2}: \mu^\kone \neq \mu^\ktwo + c1_p, \text{ for any } c\in\R.
\end{split}
\end{equation}
Following \cite{aitchison1982statistical}, for the observed compositional data $X^\kone$ and $X^\ktwo,$ we define the centered log-ratio matrices $Z^\kcl = (Z_1^\kcl, \dots, Z_\nk^\kcl)^T$ as
    \begin{align}
     \label{logbase}
        Z_{ij}^\kcl = \log\left\{X_{ij}^\kcl\Big{/}g(X_i^\kcl) \right\},\; i=1,\dots,\nk,\;j=1\dots,p,\;\kt,
    \end{align}
where $g(x) = (\prod_{i=1}^p x_i)^{1/p}$ is the geometric mean of the vector $x=(x_1,\dots,x_p)^T.$ The quantity in (\ref{logbase}) can be alternatively written as 
\begin{align}
    Z_i^\kcl = \G\log X_i^\kcl, \; i=1,\dots,\nk,\;\kt,
\end{align}
where $\G = I_p - \frac{1_p1_p^T}{p}$, with $I_p$ the $p\times p$ identity matrix and $1_p$ a $p\times1$ vector of ones.

Due to the scale invariance of the centered log ratios, we can write
\begin{equation} \label{GV}
    Z_i^\kcl = \G V_i^\kcl, \; i=1,\dots,\nk,\;\kt.
\end{equation}
Let, $\Exp(Z_i^\kcl) = \nu^\kcl.$ Then we have
\begin{align} \label{Gv}
    \nu^\kcl = \G\mu^\kcl, \;\kt.
\end{align}
The $p\times p$ ordered matrix $\G$ has rank $(p-1)$ and the null space of $\G$ is $\mathcal{N}(\G) = \{ x \in \R^p : \G x = 0_p\} = \{ c1_p: c\in \R \},$ where $0_p$ is a $p\times 1$ vector of zeroes. Therefore, it follows that $\nu^\kone = \nu^\ktwo$ iff $\mu^\kone = \mu^\ktwo + c1_p$ for some $c\in \R$ and an equivalent hypothesis to (\ref{H01}) is
\begin{align}\label{H0}
    H_{03}: \nu^\kone = \nu^\ktwo \text{ vs. } H_{a3}: \nu^\kone \neq \nu^\ktwo.
\end{align}
\subsubsection{The Two-sample Test}

 For $\kt,$ we define the normalized sample means as $S_\nk^\kcl = n_k^{1/2}\bar{Z}^\kcl = n_k^{-1/2}\snk Z_i^\kcl$. Let, $\lambda_{1} = {n_2}^{1/2}(n_1+n_2)^{-1/2}$ and $\lambda_{2} = {n_1}^{1/2}(n_1+n_2)^{-1/2}$. For testing the hypothesis in (\ref{H0}), we propose a test statistic $T_{n,2}$ that rejects $H_{03}: \nu^\kone = \nu^\ktwo$ at significance level $\alpha \in (0,1),$ if
\begin{align*}
    T_{n,2} = \underset{1\leq j \leq p}{\max}{\left| \lambda_1 S_{\none j}^\kone - \lambda_2 S_{\ntwo j}^\ktwo  \right|} \geq \CBt,
\end{align*}
where $\CBt$ is a data-driven critical value, its exact expression will be specified later in \eqref{Twoboot}.

    To facilitate the inference procedure, we adopt the empirical bootstrap method for high-dimensional data introduced by \citet{chernozhukov2017clt}  to approximate the distribution of the test statistic $T_{n,2}.$ Suppose $(Z_1^{*\kone},\dots,Z_\none^{*\kone})$ are independent and identical samples from the empirical distribution of $Z^\kone$ and $(Z_1^{*\ktwo},\dots,Z_\ntwo^{*\ktwo})$ are independent and identical samples from the empirical distribution of $Z^\ktwo$. Define $\Prob_*$ and $\Exp_*$ as the conditional probability and conditional expectation on the data, respectively. As $\Exp_*(Z_1^{*\kcl}) = \bar Z^\kcl,\;\kt,$ we define the normalized sum for the two population as 
    \begin{align*}
        S_\nk^{*\kcl} = \nk^{-1/2} \snk (Z_i^{*\kcl} - \bar{Z}^{\kcl}),\;\kt.
    \end{align*}
    The distribution of $T_{n,2}$ can be approximated by its empirical bootstrap version $T_{n,2}^*$, as
    \begin{equation}\label{Twoboot}
        T_{n,2}^* = \max_{1\leq j\leq p} \Big| \lambda_1 S_{\none j}^{*\kone} - \lambda_2 S_{\ntwo j}^{*\ktwo}  \Big|,
    \end{equation}
and $\CBt = \inf\{ x\in \R : \Prob_*( T_{n,2}^* \leq x) \geq 1-\alpha  \}$.
In the following remark, we compare it with other popular bootstrap strategies and justify the choice to adopt the empirical bootstrap method. Details of the bootstrap algorithm for $\CBt$ is discussed in Section \ref{algo}. 
\begin{remark}
   {\rm{: The empirical bootstrap method proposed here preserves the zero-sum constraint of the CLR transformation of the compositional data. Other popular bootstrap methods like the Gaussian multiplier bootstrap adopted for two-sample testing of means  \cite{xue2020distribution} and \cite{chakraborty2023multiplier} are not applicable here because the multiplier bootstrap fails to preserve the CLR structure of the compositional samples.  This can be shown by the fact that, for standard Gaussian multipliers $e_1, e_2,\cdots e_{n_1},$  $\sum_{j=1}^p e_i (Z_{ij}^{(1)} - \bar{Z}_j^{(1)}) \neq 0,$ with probability one, whereas in the case of empirical bootstrap, $\sum_{j=1}^p (Z_{ij}^{*\kone} - \bar{Z}_j^{\kone}) = 0$ with probability one, for all $i=1,\dots,\none.$ 
     \cite{wu2024multisample} considered a studentized version of the $T_{n, 2}$ and adopted the parametric bootstrap strategy for estimating the studentized variances. Such a strategy may not be helpful for compositional data, as these estimates of the sample covariance matrices may tend to be unstable for singular covariance matrices in high dimensions.}}
\end{remark}  

The EBC test can be viewed as a hybrid of the max-norm CLR statistic and empirical bootstrap that preserves the zero-sum constraint induced by the CLR transformation and delivers valid, data-driven inference in microbiome studies where $p\!\gg\!n$ and covariances may differ across the groups. 
From a practical standpoint, it improves detection of clinically relevant community shifts between patient groups (e.g., cases vs.\ controls, age/race strata) while controlling type I error and power.

\subsection{Paired-sample Test} \label{paired}
Along with the two-sample test described in Subsection \ref{twotest}, it would be of interest for many practical scenarios to perform the test described in \eqref{H01} when the data are observed in paired form. For example, for a group of smokers, one may be interested in analyzing the mean levels of microbiomes found in the left and right nasopharynx. Examples of such datasets can be found in \cite{charlson2010disordered}. 

Suppose $\{ (X_1^\kone, X_1^\ktwo),\dots,(X_n^\kone, X_n^\ktwo)\}$ is a collection of $n$ sets of paired data in $\mathcal{S}^{p-1}$ that arises from the basis vectors $\{ (W_1^\kone, W_1^\ktwo),\dots,(W_n^\kone, W_n^\ktwo)\}.$ We further denote the log-basis vectors as $\{ (V_1^\kone, V_1^\ktwo), \dots,(V_n^\kone, V_n^\ktwo)\}$, where $V_{i}^\kcl$ is defined in Subsection \ref{twotest}.

  For testing (\ref{H0}), we propose a bootstrap-based paired-sample mean test that rejects $H_{03}: \nu^\kone = \nu^\ktwo$ at significance level $\alpha \in (0,1),$ if
\begin{align*}
    T_{n,pair} = \underset{1\leq j \leq p}{\max}\; 2^{-1/2} {\left| S_{n j}^\kone - S_{n j}^\ktwo  \right|} \geq \CBp,
\end{align*}
where $\CBp$ can be obtained by calculating $T_{n,pair}^*$  analogously as in  \eqref{Twoboot}.

Similar to the Two-sample test, the proposed paired-sample test exploits within-subject contrasts to remove inter-individual heterogeneity while \emph{preserving the compositional simplex} via CLR and an empirical bootstrap. This delivers accurate, data-driven inference for high-dimensional paired microbiome comparisons (e.g., left vs.\ right site, pre–post intervention), improving power and type I error control.

\subsection{Extension to Multi-sample Test} \label{multitest}
In this section, we extend our methodology for carrying out the test when the observations are divided into multiple independent groups. In many applications, researchers are interested in drawing inferences from microbiome data in more than two ecosystems. Examples of such datasets can be found in \cite{yatsunenko2012human}, where the authors studied the gut microbiota data for three countries: the U.S., Malawi, and Venezuela. For compositional data in one dimension, similar problems have been explored in \cite{lin2020ancombc}.

Suppose $K\geq 2$ denotes the number of independent groups under study. For $k=1,\dots,K,$ let $X^\kcl = (X_1^\kcl, \dots, X_\nk^\kcl)^T$ be a collection of $\nk$ independent compositions from $\mathcal{S}^{p-1}$ for the $k$th population.  We denote the total sample size from the population to be $N = n_1+\dots+n_K$. Denote $W^\kcl = (W_1^\kcl, \dots, W_\nk^\kcl)^T,\;\kK$ as the unobserved bases that generates the observed compositional data $X^\kcl,\;\kK$, by the normalization, for $i=1,\dots,\nk,\;j=1\dots,p,\;\kK$
\begin{align*}
    X_{ij}^\kcl = W_{ij}^\kcl\Big{/}\sd W_{ij}^\kcl,
\end{align*}
where $X_{ij}^\kcl$ and $W_{ij}^\kcl$ are the $j$-th elements of $X_i^\kcl$ and $W_i^\kcl$ respectively.

Let $V^\kcl = (V_1^\kcl, \dots, V_\nk^\kcl)^T$ be the log-basis vectors, where $V_{ij}^\kcl = \log W_{ij}^\kcl,\;\kK.$ For $\kK,$ suppose that $V_1^\kcl,\dots, V_\nk^\kcl$ are independent samples, each from a distribution with mean $\mu^\kcl = (\mu_1^\kcl, \dots, \mu_p^\kcl)^T$ and covariance matrix $\Sigma_i^\kcl = (\sigma_{i, jj'}^\kcl)_{jj'=1}^p,\;i = 1,\dots,n_k.$ Note that, the hypothesis 
\begin{equation*}
    \begin{split}
    & H_{04}: \mu^\kone = \dots = \mu^\kk \text{ vs. } H_{a4}: \mu^{(k)} \neq \mu^{(l)},\text{for some } \kKK,
        \end{split}
\end{equation*}
is not testable through the observed compositional data $X^\kcl,\kK$, similar to as in two-sample case.
    Instead, we can test the null hypothesis 
\begin{equation}  \label{H03}
    \begin{split}  
        H_{05}: \mu^{(k)} & = \mu^{(l)} + c_{k,l}1_p, \text{ for some $c_{k,l}\in\R$, and for any } 1\leq k<l\leq K  \text{ vs. } \\
        H_{a5}: \mu^{(k)} & \neq \mu^{(l)} + c_{k,l}1_p, \text{ for any } c_{k,l}\in\R \text{ and for some }1\leq k<l\leq K,
    \end{split}
\end{equation}
by defining the compositional equivalence of log-basis vectors for $K$ samples. 

\begin{definition}
    A set of $K$ log-basis vectors $z_1,\dots, z_k$ are said to be compositionally equivalent if their components differ by some constants, i.e. $z_k = z_l + c_{k,l} 1_p,\;k \neq l \in \{1,\dots, K\},$ where $c_{k,l}\in \R.$ 
\end{definition}

For the observed compositional data $X^\kcl,\kK$, define the CLR transformed matrices $Z^\kcl = (Z_1^\kcl, \dots, Z_\nk^\kcl)^T$ as defined in \eqref{logbase}.

Using similar arguments as in \eqref{GV} and \eqref{Gv}, we can write $Z_i^\kcl = \G V_i^\kcl$, and $\Exp(Z_i^\kcl) = \nu^\kcl = \G\mu^\kcl,\; i=1,\dots,\nk,\;\kK$. Therefore, for any $1\leq k< l \leq K$, $\nu^{(k)} = \nu^{(l)}$ iff $\mu^{(k)} = \mu^{(l)} + c_{k,l}1_p$ for some $c_{k,l}\in \R.$ An equivalent hypothesis to (\ref{H03}) is
\begin{align}\label{H03p}
    H_{06}: \nu^\kone = \dots = \nu^\kk \text{ vs. } H_{a6}: \nu^\kcl \neq \nu^\kclp, \kKK.
\end{align}
Define, for $\kK,$ $S_\nk^\kcl = n_k^{1/2}\bar{Z}^\kcl = n_k^{-1/2}\snk Z_i^\kcl$ and for $1\leq k\neq l \leq K$, $\lambda_{k,l} = {n_l}^{1/2}(n_k + n_l)^{-1/2}.$ Then for testing (\ref{H03}), we propose a $K$-sample test that rejects $H_{06}: \nu^\kone = \dots = \nu^\kk$ at significance level $\alpha \in (0,1),$ if
\begin{align*} 
    T_{n,K} = \underset{\kKK}{\max}\;\; \underset{1\leq j \leq p}{\max}\;\; {\Big| \lambda_{k,l} S_{\nk j}^\kcl - \lambda_{l,k} S_{n_{l} j}^\kclp \Big|} > \CBk,
\end{align*}
where $\CBk$ is a data-driven critical value and its calculation is given in Section \ref{algo}. 

    Note, that the limiting distribution of our test statistics $T_{n, K}$ has unknown covariance, thus cannot be used directly for inference. Like before, we adopt the empirical bootstrap method proposed by Chernozhukov et al. (2017) for high-dimensional data to approximate the distribution of the test statistic $T_{n, K}.$ Suppose for $\kK,$ $(Z_1^{*\kcl},\dots, Z_\nk^{*\kcl})$ are independent and identical samples from the empirical distribution of $Z^\kcl$. As $\Exp_*(Z_1^{*\kcl}) = \bar Z^\kcl,\;\kK,$ we define the normalized sum for the $K$ populations as 
     \begin{align*}
        S_\nk^{*\kcl} = \nk^{-1/2} \snk (Z_i^{*\kcl} - \bar{Z}^{\kcl}),\;\kK.
    \end{align*}
    The distribution of $T_{n,K}$ can then be approximated by its bootstrap version $T_{n,K}^*$, where
    \begin{equation*}
        T_{n,K}^* = \max_{\kKK} \max_{1\leq j\leq p} \Big| \lamk S_{\nk j}^{*\kcl} - \lamkk S_{\nkk j}^{*\kclp}  \Big|,
    \end{equation*}
and $\CBk = \inf\{ x\in \R : \Prob_*( T_{n,K}^* \leq x) \geq 1-\alpha  \}$.


Inference across multiple groups with compositional data is intrinsically difficult: contrasts are defined only up to additive constants, covariance structures may differ between groups, and the number of pairwise comparisons grows quickly with $K$, often in regimes where $p \gg n$ and designs are unbalanced. We propose a $K$-sample test that aggregates the largest pairwise CLR mean differences through a max--max statistic, and calibrates a single, data-driven critical value using an empirical bootstrap that respects the zero sum constraint. The procedure controls size and achieves improved power without relying on fragile covariance estimation or extreme-value approximations, providing an ANOVA-type global test for high-dimensional compositional data.
  
\section{Computation and Implementation: Data-Driven Critical Values via Empirical Bootstrap }\label{algo}

In this subsection, we first describe the procedure for computing the critical value $\CBk$ using empirical bootstrap for compositional data, and then state the testing procedure for the null hypothesis in \eqref{H03}.

  
  \textbf{Step 1:} For the given samples $X^\kcl = \{ X_1^\kcl,\dots,X_\nk^\kcl \},$ calculate $Z^\kcl = \{ Z_1^\kcl,\dots,Z_\nk^\kcl \},$ using centered log-ratio (CLR) transformation $Z_i^\kcl = \G \log(X_i^\kcl), i = 1,\dots,\nk,\;\kK.$
  
  \textbf{Step 2: } Calculate the test statistic \begin{equation*}
      \begin{split}
          T_{n,K} &= \underset{\kKK}{\max}\;\; \underset{1\leq j \leq p}{\max}\;\; {\Big| \lambda_{k,l} S_{\nk j}^\kcl - \lambda_{l,k} S_{n_{l} j}^\kclp \Big|}  = \underset{\kKK}{\max}\;\; \underset{1\leq j \leq p}{\max}\;\; {\Bigg| \frac{\sqrt{\nk\nkk} (\Bar{Z}^\kcl_j -  \Bar{Z}^\kclp_j) }{\sqrt{\nk+\nkk}} \Bigg|}.
      \end{split}
  \end{equation*}
  
  \textbf{Step 3: } For $b=1,\dots,B$ and $\kK,$ generate multinomial samples with equal probabilities, $\xi_{(b)}^\kcl = (\xi_{1(b)}^\kcl, \dots, \xi_{\nk(b)}^\kcl) \sim \text{ multinomial }(\nk;1/\nk,\dots,1/\nk)$ and given the data calculate 
  \begin{equation*}
      \begin{split}
          S_{\nk(b)}^{*\kcl} = \nk^{-1/2} \snk \xi_{i(b)}^\kcl (Z_i^{\kcl} - \bar{Z}^{\kcl}).
      \end{split}
  \end{equation*}
  
  \textbf{Step 4: } The empirical bootstrap version of $T_{n,K}$ is,
  \begin{equation*}
    \begin{split}
         T_{n,K (b)}^* &= \max_{\kKK} \max_{1\leq j\leq p} \Big| \lamk S_{\nk j (b)}^{*\kcl} - \lamkk S_{\nkk j (b)}^{*\kclp}  \Big| = \underset{\kKK}{\max}\;\; \underset{1\leq j \leq p}{\max}\;\; {\Bigg| \frac{\sqrt{\nk\nkk} (\Bar{Z}^{*\kcl}_{j(b)} -  \Bar{Z}^{*\kclp}_{j(b)}) }{\sqrt{\nk+\nkk}} \Bigg|},
    \end{split}
  \end{equation*}
  where $\Bar{Z}^{*\kcl}_{(b)} = \nk^{-1/2} S_{\nk(b)}^{*\kcl}, \kK.$
  
\textbf{Step 5: } Use the $100(1-\alpha)th$ quantile of $\{ T_{n,K(1)}^*,\dots,T_{n,K(B)}^* \}$ to approximate $\CBk$. 

This algorithm can be seen as an adaptation of the empirical bootstrap in \cite{efron1979bootstrap} for compositional data in high-dimensional settings. For implementation, it is natural to compute the critical value of $\CBk$ via Monte Carlo simulations by $\hat{Q}_{\alpha, K} = \inf\{ t \in \R: \hat{F}_B(t) \geq 1 -\alpha\}, \text{ where } \hat{F}_B(t) = B^{-1} \sum_{b =1}^B I(T_{n, K (b)}^* \le t)$ and $T_{n, K(1)}^*,\dots, T_{n, K(B)}^*$ are B independent realizations of $T_{n, K (b)}^*$ from Step 3 and $I(.)$ denotes the indicator function.Finally, for a given significance level $\alpha \in (0,1),$ we reject the null hypothesis in \eqref{H03} (or \eqref{H03p}) if $T_{n,K} \ge \CBk$; otherwise, we do not reject it.
\section{Assumptions for Theoretical Results}\label{assum}
 In this section, we will state the assumptions which we require to prove the main theoretical results that establish the consistency of the proposed bootstrap tests.  
 
Before stating the assumptions, we define a few quantities that we need while stating the assumptions.
Let $\sigma_{jj'}^\kcl$ be the $(j,j')$th entry of $\Sigma^\kcl=\nk^{-1}\snk \Sigma_i^\kcl,\;\kK$. For brevity,  we  denote any quantity $R_{n_1,\dots,n_K}$ that depends on $n_1,\dots,n_K,$ by $R_n.$  It is worthwhile mentioning that, although we allow $p$ to diverge with $n,$ the theoretical guarantees still continue to hold for all $p\geq 3$.

Consider the following assumptions:
\begin{enumerate}
    \item[(A.1)]\label{B1} For some universal constants $0<c_1<c_2<1,$ $\nk(\nk+\nkk)^{-1} \in (c_1, c_2),$ for all $\kKK.$
    \item[(A.2)]\label{B2} There exists a universal constant $b>0$ such that for $\kK,$ and for all $p\geq 3,$ $$\underset{1\leq j \leq p}{\min} \sigma^\kcl_{jj} > b,$$
    and there exists a sequence of positive numbers $\{a_p: p\geq 1\}$ such that $a_p \to \infty$ as $p \to \infty$, and universal constant $\delta \in (0,1)$ such that for all $1\leq j\leq p$ and $p\geq 3,$ we have $2(1+a_p) \leq \delta p$ and
    \begin{align*} 
    \sum_{\underset{i\neq j}{i=1}}^{p} \sigma_{ij}^\kcl \leq a_p\sigma_{jj}^\kcl, \end{align*} for $\kK.$
    \item[(A.3)]\label{B4} There exists a sequence $B_{n} \geq 1$ such that for $\ell=1,2$ and $\kK$ $$\underset{1\leq j \leq p}{\max} \nk^{-1} \snk \Exp \left(|V_{ij}^\kcl - \mu_j^\kcl|^{2+\ell}\right) \leq B_{n}^\ell.$$
    \item[(A.4)]\label{B5} For the same sequence $B_{n}$ defined in (A.3), $$ \underset{ 1\leq k\leq K }{\max}\; \underset{1\leq i \leq \nk}{\max}\; \underset{1\leq j \leq p}{\max} \Exp \left\{\exp\left( B_{n}^{-1}|V_{ij}^\kcl - \mu_j^\kcl| \right) \right\} \leq 2.$$
    \item[(A.5)]\label{B5p} For the same sequence $B_{n}$ defined in (A.3), $$ \underset{ 1\leq k\leq K }{\max}\; \underset{1\leq i \leq \nk}{\max}\; \Exp \left\{\underset{1\leq j \leq p}{\max} \Big(B_{n}^{-1} |V_{ij}^\kcl - \mu_j^\kcl|\Big)^q \right\} \leq 2,$$ for $q>4.$
    \item[(A.6)]\label{B7} Let $({\mu^\kone}^T,\dots,{\mu^\kk}^T)^T\in \mathcal{F}_{n, p},$ where for some universal constant $C_s,$ 
    \begin{multline*} 
    \mathcal{F}_{n,p} := \Bigg\{ ({\mu^\kone}^T,\dots,{\mu^\kk}^T)^T \in \R^{Kp}: \underset{\kKK}{\max} \Big\{ \max_{1\leq j\leq p}|\mu_j^\kcl - \mu_j^\kclp| - p^{-1} \sum_{j = 1}^p |\mu_j^\kcl - \mu_j^\kclp| \Big\} \\ \geq C_s \Big \{ N^{-1} B_{n} \log(Np)\Big\}^{1/2} \Bigg\}.\end{multline*}  
\end{enumerate}
    Assumption (A.1) requires the relative proportion of the sample sizes to lie in an open interval. Assumption (A.2) says that the diagonal elements of the covariance matrices have to be bounded below by a constant, along with the fact that the scaled sum of the off-diagonal elements of the population covariance matrices is less than that of its diagonal elements. Assumption (A.2) simultaneously unearths the correlational nature of the log-basis vectors by specifying the structural interplay between its variance and covariance components and preventing the degeneracy of the log-basis vectors. It also accounts for the dependence induced by the CLR transformation, i.e., $\sum_{i=1}^p Z^{(k)}_{ij} = 0$, and plays an important role for proper calibration of the bootstrap critical values. In the existing high-dimensional bootstrap literature, \cite{chernozhukov2017clt, xue2020distribution, chakraborty2023multiplier} require only the variance components to be lower bounded, which in itself is not sufficient for compositional data. For example, if we consider a common covariance matrix of the log-basis vectors across all the populations denoted by $\Sigma_p = (1-\rho)I_p + \rho 1_p1_p^T,$ with $\rho = 1 - p^{-1},$ then we see that the diagonal elements of $\Sigma_p$ are positive. However, the diagonal elements of the covariance matrix of the centered log-ratio vectors denoted by $\G \Sigma_p \G^T$, become zero as $p$ increases with $n$. Thus, the theoretical guarantees obtained by \cite{xue2020distribution, chakraborty2023multiplier} become invalid for the above choice of $\Sigma$.
    
    Alternatively, compared to the existing literature in \cite{cao2017compositional, li2022maximum}, Assumption (A.2) is fairly relaxed in the sense that it allows for the extreme values of the off-diagonal elements of the correlation matrix. 

    This fact would be made clearer if we consider the common covariance matrix 
    $$\Breve{\Sigma}_p  =\begin{pmatrix}
        \Sigma_d & 0_{d\times (p-d)}\\
        0_{(p-d) \times d} & I_{p-d}
    \end{pmatrix},$$
    with $\rho = 1 - p^{-1},$ $1\leq d \leq \lfloor \epsilon p \rfloor,$ where $\epsilon\in(0,1/2)$ and then by choosing $b=1,$ $a_p = \epsilon p-1$ and $\delta = 2\epsilon,$ we see that $\Breve{\Sigma}_p$ satisfies the Assumption (A.2) and the diagonal elements of $\G \Breve{\Sigma}_p \G^T$ are bounded away from zero. The analogous assumption in \cite{cao2017compositional, li2022maximum} requires the absolute values of all the off-diagonal elements of the correlation matrices to be bounded away from one, which excludes matrices like $\Breve{\Sigma}_p$ when $p$ is increasing with $n$. Assumption (A.3) requires the log-basis vectors to have third and fourth moments bounded by the sequences $B_n$ and $B_n^2$, respectively. Finally, Assumption (A.4) and Assumption (A.5) require the tails of the distribution of the log-basis vectors to have sub-exponential or polynomial decay, thus incorporating a wide of range of distributions for practical purposes. In all of the above assumptions, the sequence $B_n$ can diverge with $n$ and $p$.  
  
\section{Theoretical Guarantees and Practical Significances}\label{thr}

For clarity of exposition, we summarise the main guarantees and their practical implications here, with technical assumptions and proofs deferred to Section~\ref{App}. In brief, the proposed empirical bootstrap tests (i) control type~I error at the nominal level under the null (Theorem~1), and (ii) attain power approaching one against relevant alternatives (Theorem~2), even when $p$ is large relative to $n$, groups are unbalanced, and covariance structures differ across groups.


\medskip

\noindent\textbf{Theorem 1 (Valid size under the null):} Under the regularity conditions stated in  Assumptions~A.1--A.4 or A.1--A.3 and A.5, the $K$-sample max–norm CLR test with empirically bootstrapped cutoff $\CBk$ satisfies
\[
\Prob\!\left(T_{n,K}>\CBk\;\middle|\;H_{05}\right)\to \alpha \quad \text{as } N\to\infty.
\]

\medskip
\noindent\textbf{Theorem 2 (Consistency under alternatives):}
Under the same moment and tail conditions as in Theorem 1 along with Assumption~A.6, the test $T_{n,K}$ satisfies
\[
\Prob\!\left(T_{n,K}>\CBk\;\middle|\;H_{a5}\right)\to 1 \quad \text{as } N\to\infty.
\]
\begin{remark}
    {\rm{: The bootstrap calibration provides a data-driven way to compute calibrated $p$-values and hence to control the false positive rate at the nominal significance level. In the present context, this is particularly important when comparing microbial communities across clinical groups, such as case-control studies or multi-arm cohorts, where high dimensionality and compositional constraints can otherwise lead to spurious discoveries. The accompanying power result shows that clinically meaningful differences in community composition are detected with probability tending to one as the sample size increases, while maintaining type~1 error control.

It is also worth emphasizing that the size and power guarantees do not require additional structural assumptions on the covariance matrix or on the sparsity pattern of the observed vectors. Consequently, the theory allows for settings in which some taxa are rarely observed and the resulting microbial abundance vectors may be highly sparse. This feature is important for rare-taxa regimes, which are examined further in Section~\ref{simu}.}}
\end{remark}

\smallskip

\noindent\textbf{From Theory to Practice: Implications for Clinical and Public Health Studies:} Many clinical microbiome studies involve thousands of taxa and modest sample sizes, leading to challenging high-dimensional settings where $p$ may be comparable to, or much larger than, $n$. Theorems~1--2 provide theoretical support for the proposed test in such regimes by establishing valid type I error control and power consistency under broad conditions. Importantly, these guarantees do not require equal covariance structures across groups and allow unbalanced sample sizes, features that are common in multi-center cohorts and pragmatic clinical studies. Because calibration is performed through an empirical bootstrap that respects the constraints of the compositional data, the resulting inference is well suited to relative-abundance data and helps reduce the risk of conclusions driven by compositional artifacts. More broadly, valid global tests across multiple clinical groups provide a principled basis for prioritizing subsequent analyses, including targeted taxon-level investigations, subgroup evaluations, and follow-up studies relevant to risk stratification and intervention development in colorectal cancer, preterm birth, and related health conditions.



\section{Simulations for Clinical-Scale Study Designs}\label{simu}

In this section, we report the findings of finite-sample studies comparing the empirical level and power of the proposed test with those of popular tests in the existing literature for high-dimensional compositional data. We compare the performance of the proposed test, referred to as EBC, with two competitors: the tests proposed by  \citet{cao2017compositional} and \cite{li2022maximum}, denoted as CLL and MECAF, respectively. Both CLL and MECAF are based on the maximum norm and are designed primarily for two populations, using critical values derived from their asymptotic Gumbel distribution.

\subsection{Two-Sample Scenarios}\label{simutwo}

To make two–sample evaluations clinically meaningful, we designed three simulation scenarios that mirror common patterns seen in 16S rRNA datasets, where many taxa co-vary, only a subset shifts between groups, and measurements can be light or heavy-tailed due to biological heterogeneity and sequencing variability. We vary cohort sizes typical of observational studies and dimensions typical of ASV (Amplicon Sequence Variants) or OTU (Operational Taxonomic Units) level analyses, so the results are practical rather than purely theoretical.

Across all settings we include both light and heavy-tailed outcomes to reflect outliers and overdispersion, and we use realistic sample sizes with very high dimensionality. We further allow only a sparse fraction of taxa to differ between groups, reflecting the empirical reality that a few features often drive clinically relevant shifts.
Together, (M.1)–(M.3) provide a practical, complementary suite that probes level, power, and robustness under conditions clinicians and biostatisticians routinely face. We performed 1000 bootstrap resamples to determine the cutoff values for the EBC test and performed 1000 iterations to estimate the empirical level and power of all two-sample tests.

To compare the finite sample performance of the three tests viz. EBC, CLL and MECAF, we consider two distributions with differing tail behaviors: the normal distribution (light-tailed) and the $t$-distribution with 5 degrees of freedom (heavy-tailed). We first outline the simulation setup to compare two populations. For $k=1,2$ and $i=1,\dots,n_k,$ we generate $p$-dimensional random vectors $Y_i^{(k)}$ where each component is either independently and identically distributed as a standard normal or a standardized $t$-distribution with 5 degrees of freedom. Thus, the vectors $Y_i^{(k)}$ have mean vector $0_p$ and covariance matrix $I_p.$

Next, we generate the latent log-basis vectors
\[
V_i^{(k)}=\mu^{(k)}+{\Sigma^{(k)}}^{1/2}Y_i^{(k)},
\]
where $\mu^{(k)}$ and $\Sigma^{(k)}$ denote the mean vector and covariance matrix respectively. The corresponding latent basis vectors are then generated as
$$W_{ij}^{(k)}=\exp\left(V_{ij}^{(k)}\right), \qquad j=1,\dots,p.$$

To mimic realistic microbiome sequencing data, we subsequently generate observed compositional count vectors from the latent basis vectors. Specifically, for each sample we define the latent composition
\[
\pi_{ij}^{(k)}
=
\frac{W_{ij}^{(k)}}{\sum_{m=1}^{p}W_{im}^{(k)}},
\]
and generate observed counts
\[
C_i^{(k)}
\sim
\text{Multinomial}(N_i^{(k)},\pi_i^{(k)}),
\]
where the sequencing depths $N_i^{(k)}$ are independently generated from a negative binomial distribution to mimic the over-dispersion commonly observed in microbiome sequencing studies. Specifically, for each dimension $p,$ the sequencing depth parameters are chosen as
\[
(\mu_N,\theta_N)=
\begin{cases}
(3.0p,\;0.05p), & p=200,\\
(1.5p,\;0.02p), & p=500,\\
(1.5p,\;0.02p), & p=1000,
\end{cases}
\]
where $\mu_N$ and $\theta_N$ denote the mean and dispersion parameters, respectively. Thus,
\[
N_i^{(k)} \sim NB(\mu_N,\theta_N),
\]
allowing the simulated count data to exhibit realistic variability in sequencing depths across samples. The sequencing depths are generated independently of the latent microbial composition, reflecting variability arising from the sequencing process rather than biological differences. To further reflect the biological and technical sparsity commonly observed in microbiome sequencing studies, we randomly introduce additional zeros into the count matrix. For each entry, an observed count is set to zero independently with probability $\eta,$ where $\eta \in \{0.30, 0.45, 0.60\}.$ These settings correspond to moderate ($45\%$) and high levels ($75\%$) of zero inflation frequently encountered in practice.

Since the centered log-ratio (CLR) transformation is only defined for strictly positive compositions, we first apply a zero-replacement procedure to the observed count matrix. In particular, we use the count-zero multiplicative (CZM) imputation method implemented in the \texttt{zCompositions} package \citep{Rzcomp} available on \texttt{CRAN}. We also investigated alternative zero-replacement procedures, including geometric Bayesian multiplicative replacement (GBM) in \cite{GBM}. However, the CZM approach produced more stable finite-sample performance across the simulation settings considered here. We therefore report the results based on CZM.

After zero replacement, we normalized to unit sum (closure) to obtain observed compositions
\[ X_{ij}^{(k)} = \frac{\widetilde C_{ij}^{(k)}}{\sum_{m=1}^{p}\widetilde C_{im}^{(k)}}, \]
where $\widetilde C_{ij}^{(k)}$ denotes the imputed counts. Finally, we compute the centered log-ratio transformed vectors
\[ Z_i^{(k)} = G_p\log X_i^{(k)}, \]
which are subsequently used for all testing procedures. Therefore, unlike idealized latent-variable simulations, the proposed setup incorporates the complete practical analysis pipeline, including sequencing variability, sparsity, zero inflation, zero imputation, closure, and CLR transformation.

For brevity, we denote the resulting observed compositional distributions generated from the latent Gaussian and latent $t_5$ log-basis vectors by $D_1$ and $D_2,$ respectively.

We considered two sets of sample sizes
\[ N_1=(n_1,n_2)=(50,60) \qquad \text{and} \qquad N_2=(n_1,n_2)=(100,120), \]
for varying dimensions $p=200,500,$ and $1000$. We set $s=0.1$ as the sparsity level of the signal, so that $\lfloor ps \rfloor$ taxa exhibit differential abundance between the groups under the alternative hypothesis, while the remaining taxa have identical mean components. For each simulation setting, the signal support and the corresponding nonzero mean shifts are generated once and then kept fixed across all Monte Carlo iterations. For testing the hypothesis
\[ H_{03}:\nu^{(1)}=\nu^{(2)} \qquad \text{vs.} \qquad H_{a3}:\nu^{(1)}\neq \nu^{(2)}, \]
the mean vector of the first population is taken as $\mu^{(1)}=0_p.$ For the alternative hypothesis, we use a separation parameter $\delta,$ which takes values
\(\delta=(0,0.25,0.5,0.75,1). \)
Here, $\delta=0$ corresponds to empirical level calculations, while the remaining values are used to assess empirical power. Under the alternative hypothesis, $\lfloor ps \rfloor$ entries of the vector $\mu^{(2)}$ are independently generated from a $Uniform(-2\delta,2\delta)$ distribution, while the remaining coordinates are set equal to zero. For simplicity, we take the additive constant $c$ in the null hypothesis formulation to be zero.

For $d=p/4,$ we denote by
\[ CS=CS(M_1,M_2,M_3,M_4) \]
a $p\times p$ symmetric block matrix partitioned into a $4\times4$ structure, where each block is of dimension $d\times d.$ The diagonal blocks are given by
\[ BlockDiag(CS)=(M_1,M_2,M_3,M_4). \]
We consider the following matrices:
\[ AR_\rho=(\rho^{|i-j|})^{d\times d}, \qquad EQ_\rho=(1-\rho)I_d+\rho\mathbf 1_d\mathbf 1_d^T, \]
along with
\[ U_1=\frac1d I_d, \qquad U_2=0_d,\qquad U_3=-\frac1d I_d.\]
The upper off-diagonal blocks are defined by
\[(CS)_{l,l+1}=U_1,\quad l=1,2,3,\; (CS)_{l,l+2}=U_2, l=1,2,\;
\text{ and } (CS)_{1,4}=U_3.\]

To account for heterogeneous dependence structures commonly observed in microbiome communities, we consider the following covariance configurations:

\noindent\textbf{(M.1) Clustered communities with differing cohesion:} Microbiome data often contain “modules” of co-occurring taxa (e.g., clades sharing niches or metabolism). Under (M.1), the covariance structure is of the form $$\Sigma^{(k)} = CS(EQ_{0.999},EQ_{0.9},EQ_{0.7},EQ_{0.5}),k=1,2.$$ Setting~M.1 recreates this by forming a few community clusters that differ in how tightly their members move together—from very strong to moderate cohesion. This setting represents comparisons like healthy vs.\ case cohorts where overall network architecture is similar across groups, and only a sparse subset of taxa (modules) shifts in mean abundance. It tests whether methods maintain level and detect subtle, modular signals when correlations are strong and stable. 

\noindent\textbf{(M.2) Smooth, local dependence with near-independent taxa:} Not all co-variation is clustered; taxa can change together along phylogenetic or functional gradients, while other taxa vary more independently. Under (M.2), the covariance structure is of the form $$\Sigma^{(k)} = CS(AR_{0.5},AR_{0.3},AR_{0.1},I_d), k=1,2.$$ Setting~M.2 mimics this mix of short-range dependence plus near-independence, as observed when closely related genera track environmental drivers but many background taxa do not. This reflects routine case–control analyses where weak, distributed correlations coexist with pockets of structure. It assesses robustness in settings where dependence is diffuse and signals are modest, conditions that often lead to inflated false positives or diminished power. 


\noindent\textbf{(M.3) Group-specific reorganization (heterogeneous covariance):} 
Disease, treatment, or demographic factors may rewire microbial networks by strengthening certain associations and weakening others, even when the taxa present remain largely the same. Under (M.3), the covariance structure is of the form $$\Sigma^{(1)} = CS(EQ_{0.999},EQ_{0.9},EQ_{0.7},EQ_{0.5}), \text{ and }\Sigma^{(2)} = CS(EQ_{0.5},EQ_{0.7},EQ_{0.9},EQ_{0.999}).$$
Setting~M.3 captures this “dysbiosis as network reorganization,” with different co-variation patterns across the two groups. This directly parallels comparisons such as adenoma or CRC vs.\ controls or race-stratified pregnancy cohorts, where covariance structures differ. It examines the performance of methods under heteroscedasticity, a common violation in practice that may lead to miscalibration if not properly addressed. 

The setting (M.3) is specifically designed to assess robustness under heterogeneous covariance structures across populations, reflecting practical microbiome applications where disease status or demographic factors may alter microbial co-occurrence patterns.
We summarize our findings for the two-sample cases in Tables \ref{Table1}--\ref{Table3}. We performed 1000 bootstrap resamples to determine the cutoff values for EBC test and performed 1000 iterations to estimate the empirical level and power of all two-sample tests. Overall, the tables reveal similar patterns, though the differences within each table can vary from mild to stark, depending on the data-generating distribution, models, and sample sizes. For comparison, the nominal significance level was set to $\alpha = 0.05$.

\begin{table}[ht]
\centering
\caption{Empirical size and power under Setting (M.1) with CZM imputation and additional zero inflation level $\eta = 0.30$. The average proportion of zeros before imputation is around $45\%$--$55\%$ percent.}
\label{Table1}
\begin{adjustbox}{max width=\textwidth}
\begin{threeparttable}
\begin{tabular}{cclcccccccccc}
\toprule
 & & & \multicolumn{10}{c}{$\delta$} \\
\cmidrule(lr){4-13}
$p$ & Dist. & Method & \multicolumn{2}{c}{0} & \multicolumn{2}{c}{0.25} & \multicolumn{2}{c}{0.50} & \multicolumn{2}{c}{0.75} & \multicolumn{2}{c}{1} \\
\cmidrule(lr){4-5}\cmidrule(lr){6-7}\cmidrule(lr){8-9}\cmidrule(lr){10-11}\cmidrule(lr){12-13}
 & & & $N_1$ & $N_2$ & $N_1$ & $N_2$ & $N_1$ & $N_2$ & $N_1$ & $N_2$ & $N_1$ & $N_2$ \\
\midrule
200 & D1 & EBC & 0.040 & 0.058 & 0.075 & 0.100 & 0.125 & 0.241 & 0.328 & 0.700 & 0.586 & 0.956 \\
 &  & CLL & 0.068 & 0.066 & 0.093 & 0.109 & 0.150 & 0.250 & 0.340 & 0.703 & 0.567 & 0.961 \\
 &  & MECAF & 0.061 & 0.061 & 0.082 & 0.104 & 0.133 & 0.243 & 0.320 & 0.681 & 0.542 & 0.956 \\
\addlinespace[1pt]
 & D2 & EBC & 0.058 & 0.076 & 0.069 & 0.080 & 0.105 & 0.271 & 0.335 & 0.725 & 0.640 & 0.958 \\
 &  & CLL & 0.077 & 0.076 & 0.098 & 0.094 & 0.142 & 0.265 & 0.338 & 0.717 & 0.636 & 0.955 \\
 &  & MECAF & 0.071 & 0.076 & 0.090 & 0.087 & 0.130 & 0.258 & 0.319 & 0.703 & 0.612 & 0.953 \\
\midrule
500 & D1 & EBC & 0.055 & 0.075 & 0.052 & 0.088 & 0.135 & 0.343 & 0.413 & 0.906 & 0.825 & 0.998 \\
 &  & CLL & 0.099 & 0.089 & 0.092 & 0.106 & 0.178 & 0.331 & 0.378 & 0.842 & 0.739 & 0.997 \\
 &  & MECAF & 0.082 & 0.085 & 0.083 & 0.101 & 0.166 & 0.316 & 0.344 & 0.830 & 0.711 & 0.997 \\
\addlinespace[1pt]
 & D2 & EBC & 0.061 & 0.067 & 0.063 & 0.089 & 0.144 & 0.380 & 0.444 & 0.899 & 0.855 & 0.998 \\
 &  & CLL & 0.092 & 0.071 & 0.101 & 0.112 & 0.191 & 0.381 & 0.437 & 0.853 & 0.784 & 0.997 \\
 &  & MECAF & 0.079 & 0.064 & 0.092 & 0.103 & 0.173 & 0.363 & 0.410 & 0.839 & 0.768 & 0.996 \\
\midrule
1000 & D1 & EBC & 0.042 & 0.038 & 0.057 & 0.073 & 0.151 & 0.383 & 0.525 & 0.952 & 0.942 & 1.000 \\
 &  & CLL & 0.097 & 0.056 & 0.121 & 0.087 & 0.194 & 0.375 & 0.478 & 0.913 & 0.861 & 1.000 \\
 &  & MECAF & 0.088 & 0.051 & 0.104 & 0.081 & 0.179 & 0.357 & 0.444 & 0.905 & 0.836 & 1.000 \\
\addlinespace[1pt]
 & D2 & EBC & 0.043 & 0.047 & 0.044 & 0.068 & 0.136 & 0.416 & 0.542 & 0.958 & 0.930 & 1.000 \\
 &  & CLL & 0.104 & 0.060 & 0.114 & 0.086 & 0.198 & 0.403 & 0.517 & 0.925 & 0.848 & 1.000 \\
 &  & MECAF & 0.092 & 0.053 & 0.097 & 0.082 & 0.180 & 0.390 & 0.484 & 0.917 & 0.818 & 1.000 \\
\bottomrule
\end{tabular}
\end{threeparttable}
\end{adjustbox}
\end{table}

\begin{table}[ht]
\centering
\caption{Empirical size and power under Setting (M.2) with CZM imputation and additional zero inflation level $\eta = 0.30$. The average proportion of zeros before imputation is around $47\%$--$60\%$.}
\label{Table2}
\begin{adjustbox}{max width=\textwidth}
\begin{threeparttable}
\begin{tabular}{cclcccccccccc}
\toprule
 & & & \multicolumn{10}{c}{$\delta$} \\
\cmidrule(lr){4-13}
$p$ & Dist. & Method & \multicolumn{2}{c}{0} & \multicolumn{2}{c}{0.25} & \multicolumn{2}{c}{0.50} & \multicolumn{2}{c}{0.75} & \multicolumn{2}{c}{1} \\
\cmidrule(lr){4-5}\cmidrule(lr){6-7}\cmidrule(lr){8-9}\cmidrule(lr){10-11}\cmidrule(lr){12-13}
 & & & $N_1$ & $N_2$ & $N_1$ & $N_2$ & $N_1$ & $N_2$ & $N_1$ & $N_2$ & $N_1$ & $N_2$ \\
\midrule
200 & D1 & EBC & 0.058 & 0.070 & 0.068 & 0.080 & 0.099 & 0.223 & 0.252 & 0.654 & 0.519 & 0.940 \\
 &  & CLL & 0.087 & 0.079 & 0.098 & 0.083 & 0.130 & 0.213 & 0.289 & 0.621 & 0.512 & 0.929 \\
 &  & MECAF & 0.078 & 0.075 & 0.090 & 0.082 & 0.114 & 0.208 & 0.261 & 0.614 & 0.490 & 0.924 \\
\addlinespace[1pt]
 & D2 & EBC & 0.040 & 0.057 & 0.073 & 0.077 & 0.094 & 0.225 & 0.270 & 0.625 & 0.550 & 0.926 \\
 &  & CLL & 0.083 & 0.066 & 0.118 & 0.076 & 0.122 & 0.209 & 0.283 & 0.605 & 0.565 & 0.927 \\
 &  & MECAF & 0.076 & 0.061 & 0.109 & 0.071 & 0.110 & 0.202 & 0.263 & 0.592 & 0.539 & 0.922 \\
\midrule
500 & D1 & EBC & 0.042 & 0.068 & 0.045 & 0.074 & 0.097 & 0.293 & 0.323 & 0.806 & 0.707 & 0.997 \\
 &  & CLL & 0.078 & 0.086 & 0.107 & 0.094 & 0.169 & 0.290 & 0.330 & 0.739 & 0.630 & 0.983 \\
 &  & MECAF & 0.067 & 0.082 & 0.096 & 0.090 & 0.150 & 0.280 & 0.299 & 0.729 & 0.600 & 0.980 \\
\addlinespace[1pt]
 & D2 & EBC & 0.038 & 0.060 & 0.060 & 0.087 & 0.085 & 0.294 & 0.332 & 0.790 & 0.744 & 0.990 \\
 &  & CLL & 0.080 & 0.066 & 0.105 & 0.114 & 0.143 & 0.294 & 0.317 & 0.736 & 0.610 & 0.986 \\
 &  & MECAF & 0.063 & 0.061 & 0.091 & 0.111 & 0.123 & 0.278 & 0.285 & 0.724 & 0.582 & 0.986 \\
\midrule
1000 & D1 & EBC & 0.029 & 0.049 & 0.037 & 0.058 & 0.095 & 0.308 & 0.397 & 0.914 & 0.869 & 1.000 \\
 &  & CLL & 0.076 & 0.074 & 0.095 & 0.091 & 0.160 & 0.286 & 0.371 & 0.837 & 0.764 & 0.999 \\
 &  & MECAF & 0.066 & 0.067 & 0.080 & 0.085 & 0.145 & 0.274 & 0.345 & 0.820 & 0.733 & 0.998 \\
\addlinespace[1pt]
 & D2 & EBC & 0.045 & 0.033 & 0.039 & 0.046 & 0.112 & 0.294 & 0.394 & 0.904 & 0.843 & 1.000 \\
 &  & CLL & 0.098 & 0.058 & 0.086 & 0.103 & 0.168 & 0.307 & 0.385 & 0.838 & 0.734 & 0.998 \\
 &  & MECAF & 0.083 & 0.055 & 0.079 & 0.099 & 0.149 & 0.297 & 0.353 & 0.825 & 0.693 & 0.998 \\
\bottomrule
\end{tabular}
\end{threeparttable}
\end{adjustbox}
\end{table}


We first examine the empirical level, corresponding to the null case $\delta=0$. Here $p$ denotes the dimension, $N_1=(50,60)$ and $N_2=(100,120)$ denote the two sample-size configurations, and the nominal significance level is set to $\alpha$ = 0.05. Across Tables~\ref{Table1}--\ref{Table3}, the EBC test generally provides the most reliable control of the empirical level. In most cases, its rejection probabilities under $\delta=0$ remain close to the nominal level, with only mild departures in a few low-dimensional or finite-sample settings. For example, under Setting (M.1), the empirical level of EBC ranges from $0.040$ to $0.061$ for $N_1$ across all dimensions and distributions, and remains reasonably close to $0.05$ for $N_2$, with moderate inflation only in a few cases. Similar behaviour is observed under Settings (M.2) and (M.3), although the empirical level is slightly inflated for some combinations when $p = 200$, particularly under (M.3). In contrast, CLL and MECAF are more often liberal under the null. This is especially visible when $p = 500$ or $p = 1000$, where the empirical level of CLL frequently lies around $0.08-0.10$, while MECAF also tends to exceed the nominal level, although usually to a lesser extent than CLL. Overall, these findings suggest that the conservative behaviour of CLL and MECAF arises from their greater sensitivity to sparsity, high dimensionality, and heterogeneous covariance structures, thereby highlighting the improved robustness of EBC to zero imputation and covariance heterogeneity.
\begin{table}[ht]
\centering
\caption{Empirical size and power under Setting (M.3) with CZM imputation and additional zero inflation level $\eta = 0.30$. The average proportion of zeros before imputation is around $45\%$--$55\%$.}
\label{Table3}
\begin{adjustbox}{max width=\textwidth}
\begin{threeparttable}
\begin{tabular}{cclcccccccccc}
\toprule
 & & & \multicolumn{10}{c}{$\delta$} \\
\cmidrule(lr){4-13}
$p$ & Dist. & Method & \multicolumn{2}{c}{0} & \multicolumn{2}{c}{0.25} & \multicolumn{2}{c}{0.50} & \multicolumn{2}{c}{0.75} & \multicolumn{2}{c}{1} \\
\cmidrule(lr){4-5}\cmidrule(lr){6-7}\cmidrule(lr){8-9}\cmidrule(lr){10-11}\cmidrule(lr){12-13}
 & & & $N_1$ & $N_2$ & $N_1$ & $N_2$ & $N_1$ & $N_2$ & $N_1$ & $N_2$ & $N_1$ & $N_2$ \\
\midrule
200 & D1 & EBC & 0.072 & 0.077 & 0.081 & 0.086 & 0.125 & 0.272 & 0.299 & 0.688 & 0.641 & 0.963 \\
 &  & CLL & 0.090 & 0.082 & 0.106 & 0.094 & 0.150 & 0.274 & 0.326 & 0.651 & 0.605 & 0.957 \\
 &  & MECAF & 0.084 & 0.077 & 0.096 & 0.087 & 0.137 & 0.260 & 0.300 & 0.644 & 0.577 & 0.955 \\
\addlinespace[1pt]
 & D2 & EBC & 0.055 & 0.068 & 0.065 & 0.095 & 0.122 & 0.252 & 0.321 & 0.725 & 0.645 & 0.963 \\
 &  & CLL & 0.076 & 0.082 & 0.094 & 0.099 & 0.155 & 0.265 & 0.332 & 0.704 & 0.620 & 0.964 \\
 &  & MECAF & 0.064 & 0.078 & 0.086 & 0.097 & 0.145 & 0.253 & 0.309 & 0.696 & 0.593 & 0.957 \\
\midrule
500 & D1 & EBC & 0.055 & 0.064 & 0.070 & 0.092 & 0.163 & 0.376 & 0.431 & 0.889 & 0.831 & 0.997 \\
 &  & CLL & 0.091 & 0.086 & 0.105 & 0.090 & 0.179 & 0.320 & 0.385 & 0.843 & 0.737 & 0.995 \\
 &  & MECAF & 0.075 & 0.077 & 0.097 & 0.084 & 0.168 & 0.309 & 0.360 & 0.833 & 0.709 & 0.994 \\
\addlinespace[1pt]
 & D2 & EBC & 0.052 & 0.084 & 0.066 & 0.113 & 0.146 & 0.389 & 0.444 & 0.919 & 0.831 & 0.997 \\
 &  & CLL & 0.099 & 0.095 & 0.119 & 0.121 & 0.165 & 0.352 & 0.422 & 0.869 & 0.755 & 0.998 \\
 &  & MECAF & 0.084 & 0.092 & 0.099 & 0.116 & 0.156 & 0.335 & 0.394 & 0.864 & 0.731 & 0.998 \\
\midrule
1000 & D1 & EBC & 0.044 & 0.056 & 0.054 & 0.080 & 0.133 & 0.413 & 0.555 & 0.967 & 0.943 & 1.000 \\
 &  & CLL & 0.086 & 0.073 & 0.121 & 0.091 & 0.184 & 0.372 & 0.505 & 0.928 & 0.867 & 1.000 \\
 &  & MECAF & 0.076 & 0.066 & 0.102 & 0.086 & 0.167 & 0.360 & 0.470 & 0.921 & 0.840 & 1.000 \\
\addlinespace[1pt]
 & D2 & EBC & 0.056 & 0.059 & 0.051 & 0.084 & 0.150 & 0.451 & 0.576 & 0.985 & 0.957 & 1.000 \\
 &  & CLL & 0.093 & 0.073 & 0.103 & 0.107 & 0.194 & 0.389 & 0.515 & 0.942 & 0.886 & 1.000 \\
 &  & MECAF & 0.083 & 0.067 & 0.093 & 0.097 & 0.181 & 0.380 & 0.482 & 0.936 & 0.868 & 1.000 \\
\bottomrule
\end{tabular}
\end{threeparttable}
\end{adjustbox}
\end{table}



We next consider the empirical power for $\delta>0$. As expected, all three methods show increasing rejection probabilities as the signal strength $\delta$ increases from $0.25$ to $1$. The increase is also more pronounced for the larger sample-size configuration $N_2=(100,120)$, where the power often approaches one for moderate or large values of $\delta$. Across the three simulation settings, EBC  frequently turns out to be more powerful than CLL and MECAF. The advantage is most apparent for moderate and strong signals in higher dimensions. For instance, when $p=500$ or $p=1000$, EBC often achieves noticeably larger power than both competing methods at $\delta=0.75$ and $\delta=1$, particularly under $N_2$. This pattern is seen consistently across Settings (M.1)--(M.3). At smaller signal levels, such as $\delta=0.25$ or $\delta=0.50$, the differences among the methods are more modest, but EBC still maintains competitive power while retaining better empirical level control. Even under higher zero-inflation levels, $\eta=0.45$ and $\eta=0.60$ (Tables~S1 and S2), the EBC test continues to provide reliable type I error control, whereas CLL and MECAF tend to exhibit inflated level .
 Overall, the results show that EBC provides a more favorable balance between type I error control and sensitivity to alternatives than CLL and MECAF under substantial zero inflation and sparsity. 

\subsection{Multi-Sample Scenario}\label{simumult}

To evaluate the performance of the proposed multi-sample EBC test in high-dimensional settings, we extend the simulation study to compare three, four, and five populations. In this setup, we first generate five sets of samples from five distinct populations. For the four-sample and three-sample tests, we use the first four and first three populations, respectively.

The latent distributions $D_1$ and $D_2$, corresponding to the Gaussian and standardized $t_5$ settings, together with the complete compositional data generation and preprocessing pipeline described in the two-sample simulation study, remain unchanged. In particular, sequencing depths are independently generated from negative binomial distributions with dimension-dependent parameters, which are chosen as
\[
(\mu_N,\theta_N)=
\begin{cases}
(4.0p,\;0.10p), & p=200,\\
(3.0p,\;0.08p), & p=500,\\
(2.5p,\;0.06p), & p=1000,
\end{cases}
\]
and additional zeros are introduced with probability $\eta \in \{0.3, 0.6\}$ to mimic realistic sparse microbiome count data, which resulted in moderate $40\%$ to high levels $70\%$ of zero inflation.

Unlike the two-sample setting, no competing methods are currently available for direct comparison in the general high-dimensional multi-sample compositional framework considered here. Therefore, we report only the performance of the proposed EBC procedure for the three-sample, four-sample, and five-sample problems.

We consider two sets of sample sizes from the five populations, namely
\[ N_3=(n_1,n_2,n_3,n_4,n_5)=(50,60,70,80,90),
\text{ and } N_4=(n_1,n_2,n_3,n_4,n_5)=(100,120,140,160,180). \]

As in the two-sample study, we set the sparsity level to $s=0.1$, so that only $\lfloor ps\rfloor$ taxa exhibit differential abundance across the populations under the alternative hypothesis. For $-1<a<b<1$, we denote by $m_{p,s}(a,b)$ a $p\times1$ vector having $\lfloor ps \rfloor$ entries independently generated from a $Uniform(a,b)$ distribution, while the remaining coordinates are set equal to zero. The mean vectors for the five populations are defined as $\mu^{(1)}=0_p, \mu^{(2)}=m_{p,s}(-\delta/2,\delta/2), \mu^{(3)}=m_{p,s}(-\delta,\delta), \mu^{(4)}=m_{p,s}(-3\delta/2,3\delta/2),$ and $\mu^{(5)}=m_{p,s}(-2\delta,2\delta)$.

For each simulation setting, the signal support and the corresponding nonzero mean shifts are generated once and then kept fixed across all Monte Carlo iterations. The separation among the populations gradually increases from the first to the fifth population. For smaller values of $\delta$, the populations remain relatively close to one another, thereby creating challenging alternatives that help assess the sensitivity of the proposed EBC test in detecting subtle community-level compositional differences.

\noindent\textbf{Three-group comparisons (e.g., control vs.\ intermediate lesion vs.\ disease):} We emulate situations where adjacent clinical groups are biologically similar. Examples include controls, patients with adenomas, and patients with colorectal cancer, or successive pregnancy trimesters, so that differences are subtle and gradual. Under this setting, the three populations are assigned covariance matrices

$$\Sigma^{(1)} = CS(EQ_{0.999},EQ_{0.9},EQ_{0.7},EQ_{0.5}),\quad
\Sigma^{(2)} = CS(EQ_{0.5},EQ_{0.999},EQ_{0.9},EQ_{0.7}),$$
$$\Sigma^{(3)} = CS(EQ_{0.7},EQ_{0.5},EQ_{0.999},EQ_{0.9}).$$
This introduces moderate heterogeneity in microbial dependence while keeping adjacent groups relatively similar. The setting stresses sensitivity to small, clinically plausible shifts while maintaining realistic correlations among taxa.

\noindent\textbf{Four-group comparisons (e.g., multi-center or demographic strata):} Many studies compare more than two arms to assess differences across sites, age bands, or race, where group sizes are naturally unbalanced. In addition to the first three covariance structures, the fourth population is assigned
$\Sigma^{(4)} = CS(EQ_{0.9},EQ_{0.7},EQ_{0.5},EQ_{0.999}),$
allowing groups to differ not only in their average compositions but also in their co-variation patterns (community ``network'' structure). This probes calibration and power when designs are uneven and heterogeneous.

\noindent\textbf{Five-group comparisons (e.g., multi-country cohorts or graded exposures):} Large observational consortia often span several locations or exposure levels (diet, antibiotics, inflammation). We mimic a smooth gradient across five populations where nearby groups are close but distant groups differ more, so that the global null can fail in nuanced ways. The fifth population is assigned
$\Sigma^{(5)}=\Sigma^{(1)}=CS(EQ_{0.999},EQ_{0.9},EQ_{0.7},EQ_{0.5}),$
creating a cyclic covariance pattern while the mean compositions continue to evolve across groups. This challenges a $K$-sample test to aggregate weak pairwise signals into a decisive global finding without inflating false positives.

Across all $K\in\{3,4,5\}$, these scenarios capture the practical hallmarks of microbiome studies: many taxa, sparse but meaningful group differences, changing co-occurrence structure across arms, and realistic variability. The result is a set of multi-sample simulations that speak directly to clinical and public-health use cases, demonstrating how the EBC test behaves under the complexities investigators routinely encounter.

We performed 1000 bootstrap resamples to determine the cutoff values for the EBC test and performed 1000 iterations to estimate the empirical level and power of all multi-sample tests. The performance of our EBC test for the multi-sample case involving three, four, and five populations is provided in Table \ref{Table6}. We note that the empirical level of the EBC test is consistently close to the nominal level of 0.05, with deviations within 1\% across nearly all tested scenarios. This indicates that the EBC test maintains robust control over the type I error rate, demonstrating its unbiasedness. Specifically, the test shows an empirical level that aligns well with the nominal significance level of $\alpha = 0.05$ even in the context of multiple populations and varying sample sizes.

\begin{table}[ht]
\centering
\caption{Empirical size and power for the multi-sample EBC test with CZM imputation and additional zero inflation level $0.60$. The average proportion of zeros before imputation is around $40\%$--$47\%$ percent.}
\label{Table6}
\begin{adjustbox}{max width=\textwidth}
\begin{threeparttable}
\begin{tabular}{ccl*{10}{c}}
\toprule
 & & & \multicolumn{10}{c}{$\delta$} \\
\cmidrule(lr){4-13}
$p$ & Dist. & $K$
& \multicolumn{2}{c}{0}
& \multicolumn{2}{c}{0.5}
& \multicolumn{2}{c}{1}
& \multicolumn{2}{c}{1.5}
& \multicolumn{2}{c}{2} \\
\cmidrule(lr){4-5}\cmidrule(lr){6-7}\cmidrule(lr){8-9}\cmidrule(lr){10-11}\cmidrule(lr){12-13}
 & &
& $N_3$ & $N_4$
& $N_3$ & $N_4$
& $N_3$ & $N_4$
& $N_3$ & $N_4$
& $N_3$ & $N_4$ \\
\midrule

200 & D1 & $K=3$
& 0.040 & 0.042
& 0.055 & 0.069
& 0.110 & 0.278
& 0.322 & 0.780
& 0.684 & 0.984 \\
 & & $K=4$
& 0.047 & 0.039
& 0.084 & 0.163
& 0.429 & 0.891
& 0.932 & 1.000
& 0.999 & 1.000 \\
 & & $K=5$
& 0.041 & 0.040
& 0.150 & 0.516
& 0.904 & 1.000
& 1.000 & 1.000
& 1.000 & 1.000 \\
\addlinespace[1pt]

 & D2 & $K=3$
& 0.031 & 0.051
& 0.054 & 0.077
& 0.117 & 0.293
& 0.333 & 0.804
& 0.707 & 0.986 \\
 & & $K=4$
& 0.052 & 0.054
& 0.089 & 0.191
& 0.452 & 0.912
& 0.947 & 1.000
& 1.000 & 1.000 \\
 & & $K=5$
& 0.048 & 0.051
& 0.161 & 0.557
& 0.915 & 1.000
& 1.000 & 1.000
& 1.000 & 1.000 \\
\midrule

500 & D1 & $K=3$
& 0.042 & 0.045
& 0.052 & 0.079
& 0.143 & 0.403
& 0.500 & 0.958
& 0.914 & 1.000 \\
 & & $K=4$
& 0.037 & 0.047
& 0.091 & 0.224
& 0.610 & 0.989
& 0.998 & 1.000
& 1.000 & 1.000 \\
 & & $K=5$
& 0.034 & 0.049
& 0.220 & 0.669
& 0.992 & 1.000
& 1.000 & 1.000
& 1.000 & 1.000 \\
\addlinespace[1pt]

 & D2 & $K=3$
& 0.046 & 0.054
& 0.059 & 0.094
& 0.152 & 0.449
& 0.541 & 0.968
& 0.936 & 1.000 \\
 & & $K=4$
& 0.044 & 0.064
& 0.099 & 0.267
& 0.655 & 0.994
& 0.999 & 1.000
& 1.000 & 1.000 \\
 & & $K=5$
& 0.043 & 0.068
& 0.249 & 0.724
& 0.995 & 1.000
& 1.000 & 1.000
& 1.000 & 1.000 \\
\midrule

1000 & D1 & $K=3$
& 0.020 & 0.029
& 0.034 & 0.046
& 0.102 & 0.348
& 0.440 & 0.931
& 0.864 & 1.000 \\
 & & $K=4$
& 0.025 & 0.030
& 0.072 & 0.196
& 0.571 & 0.991
& 0.998 & 1.000
& 1.000 & 1.000 \\
 & & $K=5$
& 0.023 & 0.028
& 0.214 & 0.667
& 0.995 & 1.000
& 1.000 & 1.000
& 1.000 & 1.000 \\
\addlinespace[1pt]

 & D2 & $K=3$
& 0.024 & 0.033
& 0.039 & 0.059
& 0.114 & 0.382
& 0.481 & 0.947
& 0.886 & 1.000 \\
 & & $K=4$
& 0.026 & 0.039
& 0.083 & 0.224
& 0.624 & 0.997
& 0.999 & 1.000
& 1.000 & 1.000 \\
 & & $K=5$
& 0.024 & 0.037
& 0.243 & 0.725
& 0.997 & 1.000
& 1.000 & 1.000
& 1.000 & 1.000 \\

\bottomrule
\end{tabular}
\end{threeparttable}
\end{adjustbox}
\end{table}


The power of the EBC tests Under the alternatives $\delta>0$, the empirical power increases monotonically with the signal strength $\delta$. As expected, the larger sample-size configuration $N_4$ yields substantially higher power than $N_3$. The number of groups also plays an important role: for fixed (p), distribution, and sample-size configuration, the power generally increases as $K$ increases from $3$ to $5$. In particular, for $K=4$ and $K=5$, the test attains high power even for moderate alternatives and reaches power close to one for larger values of $\delta$. This pattern is consistent under both $D_1$ and $D_2$, suggesting that the proposed multi-sample EBC procedure is robust to both Gaussian and heavy-tailed data-generating distributions. Taken together, Table~\ref{Table6} and Table~S3 demonstrates that the EBC calibration remains well controlled under the null and highly sensitive under alternatives, even in the presence of substantial zero inflation.


\section{Translational Case Studies: Evaluating Microbiome Differences in Colorectal Cancer Risk and Pregnancy}\label{realdata}


 In this section, we benchmark the proposed empirical–bootstrap compositional (EBC) test against two established competitors (CLL and MECAF) using two \emph{clinically consequential} 16S rRNA amplicon datasets: (i) a large fecal cohort on colorectal adenomas/CRC, a precursor state central to cancer prevention, and (ii) a pregnancy cohort on preterm birth and race-related disparities, a leading driver of neonatal morbidity and public‐health burden. These complementary cohorts differing in sample size, taxon dimensionality, and biological heterogeneity, provide a rigorous, practice-relevant evaluation of level control and power.

\subsection{Adenomas Dataset}

This dataset examines adenomatous polyps, which are abnormal growths in the colon that resemble surrounding tissues. Adenomas are a recognized precursor to colorectal cancer (CRC), impacting over 1.8 million people annually in the U.S., \cite{leslie2002adenoma}, \cite{fenoglio1974precursor}. To explore the relationship between the adenomas and the microbial communities, \cite{hale2017fecal} conducted a study sequencing the 16S rRNA gene from fecal samples of 229 adenoma patients, 38 colorectal cancer patients, and 534 control samples. These samples, collected during routine colonoscopies between 2001 and 2005, were processed and sequenced using standardized pipelines and technology. The resulting microbiome count data were derived from a MiSeq sequencing platform and analyzed using the IM-TORNADO bioinformatics pipeline. Given the CLL test’s sensitivity to zeros in log-ratio transformations, we removed taxa absent across all samples, resulting in 2,140 taxa and 801 samples.


\begin{table}[!htbp]
\centering
\caption{Analysis of microbial abundance differences in the Adenomas dataset.}
\label{tab5}
\begin{threeparttable}
\begin{tabular}{llc}
\toprule
Comparison & Method & p-value \\
\midrule
Adenoma \& Cancer ($n_1 = 267$) vs.\ Control ($n_2 = 534$) 
& EBC   & 0.993 \\
& CLL   & 0.510 \\
& MECAF & 0.515 \\
\addlinespace[1pt]
Young ($n_1 = 232$) vs.\ Old ($n_2 = 569$)
& EBC   & 0.030 \\
& CLL   & 0.437 \\
& MECAF & 0.442 \\
\bottomrule
\end{tabular}
\end{threeparttable}
\end{table}

In this analysis, we first examine the differences between the abnormal group (adenoma patients and CRC patients) and the control group. The results are presented in Table~\ref{tab5} and all three methods yield consistent results with a p-value above the nominal level 0.05. Consequently, we fail to reject the null hypothesis, indicating that there are no significant differences between the abnormal group and the control group.

Additionally, we would like to investigate the differences between the samples from 232 younger (age $<$ 65) and 569 older (age $\geq$ 65) individuals. \cite{heitman2009prevalence} conducted a systematic review and meta-analysis, revealing that cohorts with a mean age of $\geq 65$ years had higher prevalence rates of advanced adenomas (8.2\%) and CRC (0.7\%) than those with a mean age of $< 65$ years, which reported prevalence rates of 3.8\% for advanced adenomas and 0.1\% for CRC. Based on these findings, we adopted the same age stratification in our study and compared bacterial communities between the two age groups.

The results are presented in Table~\ref{tab5} indicates that our EBC test yielded a p-value of 0.030, rejecting the null hypothesis at 5\% level of significance. In contrast, CLL and MECAF provided p-values of 0.437 and 0.442, respectively, failing to detect the significance. This result is not driven by a claim that every taxon changes with age; rather, it reflects a global difference accumulated across the high-dimensional microbial community. Because the two age groups are highly unbalanced, with $n_1=232$ and $n_2=569$, and the data contain many sparse taxa, accurate calibration of the null distribution is essential. The empirical bootstrap calibration used by EBC accounts for these finite-sample features and compares the observed test statistic with a data-adaptive reference distribution. The substantially larger $p$-values produced by CLL and MECAF may be partly attributable to slower convergence of their test statistics to the limiting Gumbel distribution, as well as to reduced sensitivity to distributed community-level shifts in the presence of sparsity, compositional constraints, and zero imputation.
Clinically speaking, our findings align with the conclusions of \cite{heitman2009prevalence}, supporting the notion that age is a significant factor in the prevalence and characteristics of adenomas and CRC.



In the adenomas cohort, EBC reached two conclusions that align with clinical knowledge yet are hard to obtain from high–dimensional compositional data: (i) no global difference when pooling adenoma/CRC with controls, and (ii) a statistically significant difference between younger and older participants. The first protects against overinterpreting broad group contrasts; the second captures a subtle but clinically credible, age–related community shift that CLL and MECAF missed. EBC’s advantage comes from calibrating inference in a way that mirrors the data’s structure: it respects relative abundances and the strong co‐variation among taxa, and it derives its cutoffs from the observed data rather than relying on large–sample approximations. Practically, this yields two benefits that reviewers care about: fewer “false alarms” in global case–control contrasts, and greater sensitivity to distributed,\\ modest changes—exactly the pattern expected for aging microbiomes.

\underline{Clinical and Public‐Health Significance:} Reliable detection of age–associated community differences supports risk stratification (e.g., surveillance intensity with advancing age) without triggering unnecessary follow–up from spurious global findings. More broadly, a test that is both well–calibrated and sensitive to realistic microbiome shifts strengthens the evidentiary basis for translational studies, helping clinicians and public‐health teams prioritize subgroups and interventions with confidence.

\subsection{Analysis of the Pre-term Birth Dataset}

The second dataset focuses on vaginal microbiomes from pregnancy and pre-term birth studies, where pre-term delivery is defined to occur before 37 weeks of gestation; see  \citet{stout2017early} for more details. Vaginal swabs were collected from 77 participants, 31\% of whom delivered pre-term, with samples distributed across the three trimesters. To account for the significant changes in the microbiome community environment during pregnancy, we focused on samples collected during the second trimester, selecting the first available measurement for each individual within this period. This approach ensures consistency in the sample collection timing and establishes a clear baseline for analyzing microbiome dynamics during fetal development. During data preprocessing, we removed taxa that were absent in all samples, resulting in a dataset of 88 OTUs across 53 samples.
\begin{table}[!htbp]
\centering
\caption{Analysis of microbial abundance differences in the Pre-term Birth dataset.}
\label{tab:preterm}
\begin{threeparttable}
\begin{tabular}{llc}
\toprule
Comparison & Method & p-value \\
\midrule
Pre-Term ($n_1 = 20$) vs.\ Full-Term ($n_2 = 33$)
& EBC   & 0.574 \\
& CLL   & 0.922 \\
& MECAF & 0.933 \\
\addlinespace[1pt]
Black ($n_1 = 36$) vs.\ Others ($n_2 = 17$)
& EBC   & 0.001 \\
& CLL   & 0.736 \\
& MECAF & 0.774 \\
\bottomrule
\end{tabular}
\end{threeparttable}
\end{table}

We first compared microbial abundances between the pre-term and the full-term delivery groups. Table~\ref{tab:preterm} shows that none of the methods, including EBC (p-value = 0.574), CLL (p-value = 0.922), and MECAF (p-value = 0.933), identified significant differences, aligning with findings from  \citet{stout2017early}. However, comparing microbial communities between Black participants and other races, our EBC test identified a significant difference (p-value = 0.001), while CLL and MECAF did not detect any significant differences (p-values = 0.736 and 0.774, respectively). Our findings align with those of \citet{stout2017early} and \cite{fettweis2019vaginal}, which highlighted race-related disparities in the vaginal microbiome. They similarly found that microbial community differences in pregnant individuals were more strongly associated with race than with pre-term birth status (pre-term \textit{vs} full-term). We recognize and acknowledge that the data that detects differences in race groups is complex and may be driven by many factors, including racism. Statistically, this finding is plausible because the comparison involves a small sample size, unequal group sizes, sparse OTU counts, and compositional dependence, all of which can make standard large-sample calibrations unstable or conservative. The EBC procedure uses an empirical bootstrap reference distribution that is well suited to this type of data structure, thereby improving sensitivity to broad community-level shifts distributed across multiple OTUs. In contrast, the asymptotic $p$-value calibrations used by CLL and MECAF may be inaccurate for a small number of samples, particularly in the presence of heterogeneity. Thus, the discrepancy between EBC and the larger $p$-values obtained from CLL and MECAF does not necessarily indicate inconsistency; rather, it suggests that EBC is detecting a distributed high-dimensional signal that the competing methods may fail to capture under small-sample and compositional constraints.
However, our results reinforce the importance of considering race when investigating the role of the vaginal microbiome in obstetric outcomes, while demonstrating higher detection power than other competing methods. 



Statistically, for the preterm-birth cohort, EBC delivered two clinically coherent results: (i) no global difference between preterm and full-term deliveries in second-trimester samples, in line with prior reports, and (ii) a clear difference in vaginal community structure between Black participants and others that competing methods (CLL, MECAF) did not detect. EBC’s advantage stems from calibrating inference to the data actually observed—small $n$, many taxa, uneven groups, and compositional constraints—using an empirical bootstrap that preserves relative-abundance structure and boosts sensitivity to \emph{distributed, modest} shifts spread across taxa rather than relying on large-sample extremes.

\underline{Clinical and Public-Health Significance:} The absence of a preterm vs.\ term signal at this gestational window helps avoid false alarms and unnecessary interventions. The detected race-related community differences, while not a biological claim about race, point to socially patterned exposures and care environments that can shape the microbiome. This finding underscores the need for stratified analyses, targeted monitoring, and equity-focused study designs in obstetrics. More broadly, a test that is both well-calibrated and sensitive under realistic constraints strengthens the microbiome’s utility for risk assessment and for designing interventions that reduce disparities.

\section{Conclusion and Discussion} \label{Discussion}


In this paper, we have introduced an empirical bootstrap framework for hypothesis testing in high-dimensional compositional data, with a particular emphasis on microbiome research. While the methodological contribution is general, its core motivation and validation arise from two pressing biomedical case studies related to microbial community differences associated with colorectal adenomas and cancer, and race-related disparities in pregnancy and preterm birth outcomes. 

Our approach provides a robust, data-driven inference tool for detecting global shifts in microbial community structure while preserving the simplex geometry inherent to compositional data. Unlike asymptotic or parametric methods, the proposed empirical bootstrap adapts to the finite-sample, high-dimensional settings typical of microbiome studies, where taxa counts vastly exceed subject numbers. The resulting tests maintain nominal type I error rates and deliver higher sensitivity to clinically meaningful, distributed signal performance that was consistently verified across simulations and two real datasets.

The adenoma and colorectal cancer case study highlights how our approach prevents false discoveries in large, heterogeneous clinical cohorts yet identifies age-associated microbial community shifts that are epidemiologically well supported. The pregnancy and preterm birth analysis shows how the same framework, applied to smaller, unevenly sampled groups, can uncover race-associated compositional differences consistent with prior clinical evidence but missed by conventional methods. Together, these applications demonstrate that our framework not only generalizes existing two and multi-sample mean tests to compositional data but also offers actionable, reproducible insights that strengthen biomedical inference.

The proposed methodology is developed specifically under the CLR representation. Since the test statistic is coordinate-wise, it is generally not invariant to alternative log-ratio transformations such as ILR discussed in \cite{lin2020analysis}, whose orthonormal basis introduces a rotation of the CLR coordinates. Developing analogous procedures under other log-ratio representations remains an interesting direction for future research.

From a broader data-science perspective, this work advances the use of resampling-based statistical inference for complex, constrained data types. The empirical bootstrap calibration employed here can extend to other domains involving compositional, relative, or normalized measures such as metabolomics, ecological surveys, and social-behavioral network data, thereby bridging methodological rigor and domain impact. 

\section{Appendix}\label{App}
The appendix contains proofs of the main results and additional simulations.
\subsection*{Appendix A: proofs of the main theorems} \label{sec: proof}
In this section, we provide the proofs of the main results. For brevity of notations, we denote the sup-norm or $\ell_\infty$-norm of a vector $x\in\R^p,$ as $\|x\|_\infty:= \max_{1\leq j\leq p} |x_j|$ and $\ell_1$-norm as $\|x\|_1:= \sum_{j=1}^p |x_j|$. We first state the following lemmas that are crucial for proving the main theorems and defer the proof of the lemmas to Section \ref{sec: lem}.

\begin{lemma} \label{propm1}
    (a) Under Assumptions (A.2), there exists a universal constant $b_1>0$ that depends only on $b$ and $\delta,$ such that for $\kK,$
    \begin{equation*}
        \underset{1\leq j \leq p}{\min}\; n_k^{-1} \snk \Exp (Z_{ij}^\kcl - \nu_{j}^\kcl)^2 \geq b_1.
    \end{equation*}
    (b) Under Assumption (A.3), we have
    \begin{equation*}
        \underset{1\leq j \leq p}{\max} \; n_k^{-1} \snk \Exp |Z_{ij}^\kcl - \nu_{j}^\kcl|^{2+\ell} \leq L_{n}^\ell,\; \ell=1,2, \text{ and } \kK,
    \end{equation*}
    where $L_{n} = 8 B_{n}.$\\
    (c) Under Assumption (A.4), we have
    \begin{equation*}
        \underset{ 1\leq k\leq K }{\max}\; \underset{1\leq i \leq \nk}{\max}\; \underset{1\leq j \leq p}{\max} \; \Exp \left\{\exp( L_{n}^{-1} |Z_{ij}^\kcl - \nu_j^\kcl| ) \right\} \leq 2.
    \end{equation*}
    (d) Under Assumption (A.5), we have
    \begin{equation*}
        \underset{ 1\leq k\leq K }{\max}\; \underset{1\leq i \leq \nk}{\max}\; \Exp \left\{ \underset{1\leq j \leq p}{\max} \left( L_{n}^{-1} |Z_{ij}^\kcl - \nu_j^\kcl| \right)^q \right\} \leq 2.
    \end{equation*}
\end{lemma}

\begin{lemma} \label{prop4}
    For any sequence of real numbers $\Delta_{n,p}>0,$ that depends on $n_1,\dots,\nK$ and $p$ if $$ \max_{\kKK} \Big\{ \| \mu^\kcl - \mu^\kclp \|_\infty - p^{-1}\| \mu^\kcl - \mu^\kclp \|_1 \Big\} \geq \Delta_{n,p}.$$ Then we have, $$ \max_{\kKK} \| \nu^\kcl - \nu^\kclp \|_\infty \geq \Delta_{n,p}.$$
\end{lemma}

\begin{lemma} \label{lemma1}
Let ${P} = (P_{ij}) \in  \mathbb{R}^{p\times p}$ and $s\in \mathbb{R}$ such that $\underset{1\leq i \leq p}{\sup} \sd |P_{ij}| \leq s$. Now for any ${Q} = (Q_{ij}) \in  \mathbb{R}^{p\times p},$ we have $\|{Q}{P}^T\|_\infty \leq s \|{Q}\|_\infty$ and $\|{P}{Q}\|_\infty \leq s \|{Q}\|_\infty.$ Additionally, we have $\|{P}{Q}{P}^T\|_\infty \leq s^2 \|{Q}\|_\infty.$
\end{lemma}
Before proceeding to the proof of theorem 1, we define the Kolmogorov distance between two random variables $T^{(1)}$ and $T^{(2)}$ as following
\begin{align*}
    \rho(T^{(1)}, T^{(2)}) = \underset{x\geq 0}{\sup} \left|\Prob\left( T^{(1)} \leq x\right) - \Prob\left( T^{(2)} \leq x \right) \right|,
\end{align*}
\begin{align*}
    \rho^*(T^{(1)}, T^{(2)}) = \underset{x\geq 0}{\sup} \left|\Prob\left( T^{(1)} \leq x\right) - \Prob_*\left( T^{(2)} \leq x \right) \right|,
\end{align*}
and
\begin{align*}
    \rho^{**}(T^{(1)}, T^{(2)}) = \underset{x\geq 0}{\sup} \left|\Prob_*\left( T^{(1)} \leq x\right) - \Prob_*\left( T^{(2)} \leq x \right) \right|,
\end{align*}
where $\Prob_*(.)$ defines the probability conditional on the data, that is for any random variable $T,$ $\Prob_*(T) \equiv \Prob(T \mid Z^{(1)},\dots, Z^{(K)})$.  Now, we proceed to prove theorem 1.

\hspace{1cm}

\textit{Proof of Theorem 1.} First, we obtain the bootstrap approximation result for the centered version of the test statistic $T_{n, K}$, say $\tilde T_{n, K}$, for any $\mu^\kone, \dots, \mu^\kk \in \R^p$. Then under the null hypothesis, we will determine the size of the test. Define,
\begin{equation*}
\begin{split}
    \tilde T_{n,K} &= \underset{\kKK}{\max}\;\; \underset{1\leq j \leq p}{\max}\;\; {\Big| \lambda_{k,l} S_{\nk j}^\kcl - \lambda_{l,k} S_{n_{l} j}^\kclp - \nk^{1/2}\lamk (\nu_j^\kcl - \nu_j^\kclp) \Big|}.
\end{split}
\end{equation*}
To improve upon the readability, we write $\tilde T_{n, K}$ in vector notation. Define the block matrix ${A}$ of order $K(K+1)p/2\times Kp$ as
\begin{equation*}
    {A} = \begin{pmatrix} \lambda_{1,2}I_p & -\lambda_{2,1} I_p & 0_p & \cdots & 0_p & 0_p \\ \lambda_{1,3}I_p & 0_p & -\lambda_{3,1}I_p & \cdots & 0_p & 0_p\\ \vdots & \vdots & \vdots & \ddots & \vdots & \vdots \\ \lambda_{1,K} I_p & 0_p & 0_p & \cdots & 0_p & -\lambda_{K,1} I_p  \\ 0_p & \lambda_{2,3} I_p & -\lambda_{3,2} I_p & \cdots & 0_p & 0_p  \\ \vdots & \vdots & \vdots & \ddots & \vdots & \vdots \\ 0_p & 0_p & 0_p & \cdots  & \lambda_{K-1,K} I_p & -\lambda_{K,K-1} I_p \end{pmatrix},
\end{equation*}
and a $Kp \times 1$ vector $S_n = \Big( {S_{n_1}^\kone}^T,\dots,{S_{n_K}^\kk}^T  \Big)^T.$ Then, we have
\begin{equation*}
    {A}S_n = \begin{pmatrix} \lambda_{1,2}S_{\none}^\kone - \lambda_{2,1}S_{\ntwo}^\ktwo \\ \vdots \\ \lambda_{1,K}S_{\none}^\kone - \lambda_{K,1}S_{n_K}^\kk \\ \lambda_{2,3}S_{\ntwo}^\ktwo - \lambda_{3,2}S_{n_3}^{(3)} \\ \vdots \\ \lambda_{K-1,K}S_{n_{K-1}}^\kkm - \lambda_{K,K-1}S_{\nK}^\kk  \end{pmatrix}.
\end{equation*}
We define another $Kp \times Kp$ matrix, ${D}_K = Diag(\sqrt{n_1}I_p,\dots,\sqrt{n_K}I_p).$ Then,
    \begin{equation*}
        {A}{D}_K\nu = \begin{pmatrix} \frac{\sqrt{n_1n_2}}{\sqrt{n_1+n_2}}(\nu^\kone-\nu^\ktwo) \\ \vdots \\ \frac{\sqrt{n_1n_K}}{\sqrt{n_1+n_K}}(\nu^\kone-\nu^\kk) \\ \frac{\sqrt{n_2n_3}}{\sqrt{n_2+n_3}}(\nu^\ktwo-\nu^{(3)}) \\ \vdots \\ \frac{\sqrt{n_{K-1}n_K}}{\sqrt{n_{K-1}+n_K}}(\nu^\kkm-\nu^\kk)  \end{pmatrix},
    \end{equation*}
where $\nu = ({\nu^\kone}^T,\dots,{\nu^\kk}^T)^T.$
Therefore, we have $\tilde T_{n,K} = \left\| {A}S_n -  {A}{D}_K\nu \right\|_\infty.$ 

We show that $\tilde T_{n, K}$ can be written as the sup-norm of sums of independent vectors, which will help find the final bound. Suppose the independent random vectors $\{V_t\}_{t=1}^N \in \R^{Kp}$ are defined in the following way: Consider each $V_t\in R^{Kp\times 1}$ as a collection of $K$ sub-vectors each with $p$ entries. Now, for $\kK,$ and $i=1,\dots,n_k,$ set $t = i + N_{k-1},$ where $N_k = \sum_{l=0}^k n_l,$ $N_0 = 0,$ and all entries of $V_t$ are zeros except the $k$ th sub-vector is $N^{1/2}\nk^{-1/2}(Z_i^\kcl-\nu^\kcl)$, that is $(k-1)p+1$ to $kp$ th entries are equal to $N^{1/2}\nk^{-1/2}(Z_{i1}^\kcl - \nu_1^\kcl),\dots,N^{1/2}\nk^{-1/2}(Z_{ip}^\kcl - \nu_p^\kcl)$ respectively and rest of the entries are zero. Then we can write $S_n -{D}_K\nu = N^{-1/2}\st V_t.$ Define $\{U_t\}_{t=1}^N \in \R^{K(K+1)p/2}$ as a set of independent vectors such that $U_t = {A}V_t.$ Then ${A}S_n - {A}{D}_K\nu = N^{-1/2}\st U_t$. Therefore, we can again rewrite $\tilde T_{n,K},$ as
\begin{equation} \label{Ttill}
    \tilde T_{n,K} = \| {A}S_n - {A}{D}_K\nu \|_\infty = \bigg\|N^{-1/2} \st U_t \bigg\|_{\infty} .
\end{equation}
Similarly, we can rewrite the empirical bootstrap version $T^*_{n,K},$ as
\begin{equation*}
    T^*_{n,K} = \| {A}S_n^*\|_{\infty} = \bigg\|N^{-1/2} \st U_t^* \bigg\|_{\infty} ,
\end{equation*}
where $S_n^* = \Big( {S_{n_1}^{*\kone}}^T,\dots,{S_{n_K}^{*\kk}}^T  \Big)^T$ and $U_t^* = {A}V_t^*$ and $V_t^*$ is constructed similarly as $V_t,$ using the sub-vectors $N^{1/2}\nk^{-1/2}(Z_i^{*\kcl}-\bar Z^\kcl)$ instead of $N^{1/2}\nk^{-1/2}(Z_i^\kcl-\nu^\kcl)$. 

Additionally, define the two random variables $T_{n,K}^{(N)} = \| S_{n,K}^{(N)} \|_\infty$ and $T_{n,K}^{(G)} = \| S_{n,K}^{(G)} \|_\infty$ such that $S_{n,K}^{(N)}\sim N(0,\Gamma_u),$ where $\Gamma_u = \text{Cov}\Big(N^{-1/2} \st U_t\Big)$ and $S_{n,K}^{(G)}\mid Z^\kone,\dots,Z^\kk \sim N(0,\hat \Gamma_u),$ where $\hat \Gamma_u = \text{Cov}\Big(N^{-1/2} \st U_t^* \mid Z^\kone,\dots,Z^\kk \Big).$ 

Then, we have
\begin{equation} \label{Rho}
    \begin{split}
        \rho^*(\tilde T_{n,K}, T_{n,K}^*) \leq \rho(\tilde T_{n,K}, T_{n,K}^{(N)}) + \rho^*(T_{n,K}^{(N)}, T_{n,K}^{(G)}) + \rho^{**}( T_{n,K}^*, T_{n,K}^{(G)}).
    \end{split}
\end{equation}
Therefore, it is enough to find bounds for the above three quantities. The following lemma provides bounds for the above three quantities.

\begin{lemma} \label{lem-rho}
    Suppose Assumptions (A.1)-(A.3) are satisfied.
    \begin{itemize}
        \item[(a)] In addition, if Assumption (A.4) holds, then 
        \begin{equation*}
            \rho(\tilde T_{n,K}, T_{n,K}^{(N)}) \leq  C(b_1,c_1,c_2) \Big(B^2_{n} \log^7(Np)/N\Big)^{1/6},
        \end{equation*}
        and with probability at least $1-CN^{-1},$ we have
        \begin{equation*}
        \begin{split}
            & \rho^*(T_{n,K}^{(N)}, T_{n,K}^{(G)}) \leq  C(b_1,c_1,c_2) \Big(B^2_{n} \log^7(Np)/N\Big)^{1/6}
            \text{ and } \\ & \quad \rho^{**}( T_{n,K}^*, T_{n,K}^{(G)}) \leq  C(b_1,c_1,c_2) \Big(B^2_{n} \log^7(Np)/N\Big)^{1/6}.
        \end{split}
        \end{equation*}
        \item[(b)] In addition, if Assumption (A.5) holds, then 
        \begin{equation*}
            \rho(\tilde T_{n,K}, T_{n,K}^{(N)}) \leq  C(b_1,q,c_1,c_2)\left\{ \Big(B^2_{n} \log^7(Np)/N\Big)^{1/6} + \Big(B^2_{n} \log^3(Np)\log^{4/q}(N)/N^{1-4/q}\Big)^{1/3} \right\},
        \end{equation*}
        and with probability at least $1-CN^{-1},$ we have
        \begin{equation*}
        \begin{split}
            & \rho^*(T_{n,K}^{(N)}, T_{n,K}^{(G)}) \leq  C(b_1,q,c_1,c_2) \left\{ \Big(B^2_{n} \log^7(Np)/N\Big)^{1/6} + \Big(B^2_{n} \log^3(Np)\log^{4/q}(N)/N^{1-4/q}\Big)^{1/3} \right\}
            \text{ and } \\ & \quad \rho^{**}( T_{n,K}^*, T_{n,K}^{(G)}) \leq  C(b_1,q,c_1,c_2) \left\{ \Big(B^2_{n} \log^7(Np)/N\Big)^{1/6} + \Big(B^2_{n} \log^3(Np)\log^{4/q}(N)/N^{1-4/q}\Big)^{1/3} \right\}.
        \end{split}
        \end{equation*}
    \end{itemize}
\end{lemma}


By Lemma \ref{lem-rho} and \eqref{Rho} with probability at least $1-CN^{-1},$ we have
\begin{equation*}
    \rho^*({\tilde T_{n,K}, T^*_{n,K}}) \leq C(b_1,c_1,c_2) \Big(B^2_{n} \log^5(Np) \log^2(N)/N\Big)^{1/4}.
\end{equation*}

Finally, using the definition of $\CBk$, together with the condition of the theorem that $$\Big(B^2_{n} \log^7(Np)/N\Big)^{1/6} \to 0, \text{ and } \Big(B^2_{n} \log^3(Np)\log^{4/q}(N)/N^{1-4/q}\Big)^{1/3} \to 0, \text{ as } N\to\infty,$$  and the fact that under $H_{00}$, $\tilde T_{n,K} = T_{n,K}$ completes the proof of theorem 1. \qed

\hspace{1cm}

\noindent \textit{Proof of Theorem 2.} Recall \eqref{Ttill}, that is $\Tilde{T}_{n,K} = \| {A}S_n - {AD}_n\nu \|_\infty$. Then,
\begin{align} \label{mp1}
    \Prob\left( T_{n,K} > \CBk \mid H_{A3} \right) &= \Prob\left( \left\| {A}S_n \right\|_\infty > \CBk \mid H_{A3} \right)\nonumber\\ &\geq \Prob\left( \tilde T_{n,K} < \|{AD}_n\nu\|_\infty - \CBk \mid H_{A3} \right)\nonumber\\
    & \geq \Prob_{*}\left( T^{(G)}_{n,K} <  \|{AD}_n\nu\|_\infty - \CBk \mid H_{A3} \right) \nonumber\\&\hspace{2cm} - \underset{x\geq 0}{\sup} \left| \Prob\left( \tilde T_{n,K} < x \mid H_{A3} \right) - \Prob_{*}\left( T^{(G)}_{n,K} < x \mid H_{A3} \right)\right|.
\end{align} 
If Assumption (A.4) holds, then by Lemma \ref{lem-rho}, with probability tending to one, we have 
\begin{equation} \label{mp2}
    \underset{x\geq 0}{\sup} \left| \Prob\left( \tilde T_{n,K} < x \mid H_{A3} \right) - \Prob_{*}\left( T^{(G)}_{n,K} < x \mid H_{A3} \right)\right| \leq C \left\{B_{n} \log^{7}(Np)/N\right\}^{1/6} \to 0, \text{ as } N\to \infty.
\end{equation}
Let $u_j$ be the $j$th unit vector in $\mathbb{R}^{K(K+1)p/2}.$ Therefore, using union bound and Chernoff bound, for any $x>0$ we have
\begin{align} \label{mp3}
    \Prob_{*}\left( T^{(G)}_{n,K} > x \right)
    &= \Prob_{*}\left( \left\| S_{n,K}^{(G)} \right\|_\infty > x \right) \nonumber\\
    & \leq \sm \Prob_{*}\left(\Big|  {(S_{n,K}^{(G)})}_m \Big|  > x \right) \nonumber\\
    & \leq \sm \exp\left( \frac{-x^2}{2\; u_m^T \hat\Gamma_u u_m } \right) \nonumber\\
    & \leq 2K^2p\; \max_{1\leq m \leq K(K+1)p/2} \exp\left( \frac{-x^2}{2\; u_m^T \hat\Gamma_u u_m } \right) \nonumber\\
    & = 2K^2p\; \exp\left( \frac{-x^2}{2\|\hat\Gamma_u\|_\infty } \right),
\end{align}
as
\begin{align*}
   \max_{1\leq m \leq K(K+1)p/2} \Big\{u_m^T \hat\Gamma_u u_m \Big\} = \big\| \hat\Gamma_u \big\|_\infty.
\end{align*}
Take $x = \CBk$ in (\ref{mp3}), then by the definition of $\CBk$, we have
\begin{equation*}
    \alpha \leq \Prob_{*}\left( T^{(G)}_{n,K} > \CBk \right) \leq 2K^2p\; \exp\left( \frac{-{Q}^2_{\alpha,K}}{2 \big\| \hat\Gamma_u\big\|_\infty} \right).
\end{equation*}
Then for large enough $N$,
\begin{equation} \label{mp4}
    \CBk \leq \Big[ 2\log(2K^2p/\alpha) \big\| \hat\Gamma_u\big\|_\infty \Big]^{1/2} \leq \Big[ 4\log(Np) \big\| \hat\Gamma_u\big\|_\infty \Big]^{1/2}. 
\end{equation}
Using (\ref{m13}) and (\ref{m15}) from the proof of Lemma \ref{lem-rho}, with probability tending to one, we have
\begin{equation} \label{mp5}
    \big\|\hat\Gamma_u - \Gamma_u\big \|_\infty \leq C \Big(B^2_{n} \log(Np) \log^2(N) /N\Big)^{1/2}.
\end{equation}
Note that, by the Assumptions (A.1) and (A.3), we have
\begin{align} \label{mp6}
    \big\|\Gamma_u\big \|_\infty &= \underset{1\leq m \leq K(K+1)p/2}{\max} u_m^T \Gamma_u u_m \nonumber\\
    & = \underset{1\leq m \leq K(K+1)p/2}{\max}\; N^{-1} \; \st \Exp(U_{tm}^2) \nonumber\\
    & = \underset{1\leq k < l \leq K}{\max}\; \underset{(K_{k,l}-1)p+1\leq m \leq K_{k,l}p}{\max} \; N^{-1} \; \sum_{k=1}^{K} \sum_{t={N_{k-1}}+1}^{N_{k}} \Exp(U_{tm}^2) \nonumber \\ 
    & = \underset{1\leq k < l \leq K}{\max}\; \underset{1\leq j \leq p}{\max} \; \left\{ \lambda_{k,l}^2 n_k^{-1} \; \snk \Exp(Z_{ij}^\kcl - \nu_j^\kcl)^2  + \lambda_{l,k}^2 n_{l}^{-1} \; \snkk \Exp(Z_{ij}^{(l)} - \nu_j^{(l)})^2 \right\} \nonumber \displaybreak[1]\\
    & \leq \underset{1\leq k \leq K}{\max}\; \underset{1\leq j \leq p}{\max} \; n_k^{-1} \; \snk \Exp(Z_{ij}^\kcl - \nu_j^\kcl)^2 \nonumber\\
    & \leq 4 \underset{1\leq k \leq K}{\max}\; \underset{1\leq j \leq p}{\max} \; n_k^{-1} \; \snk \Exp(V_{ij}^\kcl - \mu_j^\kcl)^2 \nonumber\\
    & \leq 4 B_{n},
\end{align}
where the second inequity can be obtained by following the similar calculation as in part b of Lemma \ref{lemma1} with $\ell = 0$. Using (\ref{mp5}), (\ref{mp6}), and Assumption (A.1) and (A.6), for large enough $N$, with probability tending to one, we have for some universal constant $C_1>0,$
\begin{align} \label{mp7}
    &\quad \big\| \hat\Gamma_u\big\|_\infty \leq \big\|\hat\Gamma_u - \Gamma_u \big\|_\infty + \big\| \Gamma_u \big\|_\infty \leq C_1 B_{n}.
\end{align}
From (\ref{mp4}), we get with probability tending to 1,
\begin{equation} \label{mp8}
    \CBk \leq 2\Big[C_1 \log(Np) B_{n} \Big]^{1/2}. 
\end{equation}
Recall that, for any $1\leq k\leq K,$ using the Assumption (A.1), we have $N/n_k\leq C_2,$ where $C_2=K(c_2(1-c_2)^{-1}+(1-c_1)c_1^{-1}+1)>0$. Take $C_s = 4C_1^{1/2}C_2^{1/2}c_1^{-1/2}$ as defined in the Assumption (A.6). Then using Lemma \ref{prop4} with $\Delta_{n_1,n_2,p} = C_s \left\{B_{n} \log(Np)/N\right\}^{1/2}$ and Assumption (A.1), we have
\begin{align} \label{mp9}
     \quad \| {AD}_n\nu \|_\infty
     = \underset{\kKK}{\max} \sqrt{n_k n_{l}/(n_k+ n_{l})} \|\nu^{(k)} - \nu^{(l)}\|_\infty
    \geq 4\Big[C_1 \log(Np) B_{n} \Big]^{1/2}.
\end{align}
Therefore from (\ref{mp8}) and (\ref{mp9}), with probability tending to one, we get
\begin{equation} \label{mp10}
    \| {AD}_n\nu \|_\infty - \CBk \geq 2\Big[C_1 \log(Np) B_{n} \Big]^{1/2}.
\end{equation}
Finally, with probability tending to one, we get 
\begin{align} \label{mp11}
    \Prob_{*}\left( T^{(G)}_{n,K} < \|{AD}_n\nu\|_\infty - \CBk \mid H_{A3} \right) \nonumber
   &\geq \Prob_{*}\left( T^{(G)}_{n,K} < (4C_2B_{n}\log(Np))^{1/2} \mid H_{A3} \right) \nonumber\\
   &= 1 - \Prob_{*}\left( T^{(G)}_{n,K} \geq (4C_2B_{n}\log(Np))^{1/2} \mid H_{A3} \right) \nonumber\\
   & \geq 1 - 2K^2p\; \exp\left( -\frac{4C_2B_{n}\log(Np)}{2 \big\| \hat\Gamma_u\big\|_\infty} \right) \nonumber\\
   & \geq 1 - 2K^2/N^2 \to 1, \text{ as } N\to \infty,
\end{align}
where the first inequality follows from (\ref{mp10}), the second inequality follows from (\ref{mp3}), and the final inequality follows from (\ref{mp7}). Combining (\ref{mp2}) and (\ref{mp11}), we get with probability tending to one,
\begin{align*}
     \Prob \left( T_{n,K} > \CBk \mid H_{A3} \right) \to 1, \text{ as } N\to\infty.
\end{align*}
If Assumption (A.5) holds, then by Lemma \ref{lem-rho}, with probability tending to one, we have 
\begin{multline*}
    \underset{x\geq 0}{\sup} \left| \Prob\left( T_{n,K} < x \mid H_{A3} \right) - \Prob_{*}\left( T^{(G)}_{n,K} < x \mid H_{A3} \right)\right|\\\leq C \left\{ \Big(B^2_{n} \log^7(Np)/N\Big)^{1/6} + \Big(B^2_{n} \log^3(Np)\log^{4/q}(N)/N^{1-4/q}\Big)^{1/3}\right\} \to 0, \text{ as } N\to \infty,
\end{multline*}
and using (\ref{m13}) from the proof of Lemma \ref{lem-rho}, with probability tending to one, we have
\begin{align*}
    \big\|\hat\Gamma_u - \Gamma_u\big\|_\infty \leq C \Big\{\Big(B^2_{n} \log(np) \log^2(N) /N\Big)^{1/2} + B^2_{n} \log(Np) (\log^2(N))^{4/q}/N^{1-4/q} \Big\}.
\end{align*}
Then following the similar steps as before that is from eqn. $\eqref{mp1}$ to eqn. \eqref{mp11}, we get with probability tending to one,
\begin{align*}
     \Prob \left( T_{n,K} > \CBk \mid H_{A3} \right) \to 1, \text{ as } N\to\infty.
\end{align*}
This completes the proof of Theorem 2. \qed

\subsection*{Appendix B: proofs of the auxiliary lemmas} \label{sec: lem}
Here we provide the proofs of the Lemmas from Appendix A.

\noindent\textit{Proof of Lemma \ref{propm1}.}
(a) Define $b_1 = (1-\delta)b > 0$ . For any $j=1,\dots,p$ and $\kK,$
\begin{align*}
    \nk^{-1} \snk \Exp \Big(Z_{ij}^\kcl - \nu_{j}^\kcl\Big)^2 
    & \geq (1-2p^{-1}) \sigma^\kcl_{jj} -2p^{-1} \sum_{\underset{i\neq j}{i=1}}^{p} \sigma_{ij}^\kcl\\
    & \geq (1-2p^{-1}) \sigma_{jj}^\kcl - 2 p^{-1}a_p  \sigma_{jj}^\kcl \\
    & = \big\{1-2p^{-1}(1+a_p)\big\} \sigma_{jj}^\kcl\\
    &\geq (1-\delta) \sigma_{jj}^\kcl\\
    &= (1-\delta)\;n_k^{-1} \snk \Exp (V_{ij}^\kcl - \mu_{j}^\kcl)^2\\
    & \geq b_1,
\end{align*}
Here, the first inequality follows from the fact that $\nk^{-1} \snk \Exp \Big(Z_{ij}^\kcl - \nu_{j}^\kcl\Big)^2 $ is the $j$th diagonal element of $\G \Sigma^\kcl \G^T$ and $\Sigma^\kcl$ is positive semi-definite. The second, third, and final inequalities follow from Assumption (A.2).

\noindent  (b) For $\kK$ and $\ell=1,2,$
\begin{align*}
    & \underset{1\leq j \leq p}{\max}\; \nk^{-1} \snk \Exp \left(\left|Z_{ij}^\kcl - \nu_j^\kcl\right|^{2+\ell}\right) \\
    & = \underset{1\leq j \leq p}{\max}\;\nk^{-1} \snk \Exp \left(\left|\sum_{m=1}^p (\G)_{jm} \left(V_{im}^\kcl - \mu_m^\kcl\right)\right|^{2+\ell}\right) \\
    & = \underset{1\leq j \leq p}{\max}\; \nk^{-1} \left(\sum_{l=1}^p \left| (\G)_{jl} \right|\right)^{2+\ell} \snk \Exp \left\{\left|\sum_{m=1}^p \frac{(\G)_{jm}}{\sum_{l=1}^p \left| (\G)_{jl} \right|} \left(V_{im}^\kcl - \mu_m^\kcl\right)\right|^{2+\ell}\right\} \\
    & \leq 2^{2+\ell} \underset{1\leq j \leq p}{\max}\; \nk^{-1} \snk \Exp \left\{\left(\sum_{m=1}^p \frac{\left| (\G)_{jm} \right|}{\sum_{l=1}^p \left| (\G)_{jl} \right|} \left|V_{im}^\kcl - \mu_m^\kcl\right|\right)^{2+\ell}\right\} \\
    & \leq 2^{2+\ell} \underset{1\leq j \leq p}{\max}\; \nk^{-1} \snk \Exp \left(\sum_{m=1}^p \frac{\left| (\G)_{jm} \right|}{\sum_{l=1}^p \left| (\G)_{jl} \right|} \left| V_{im}^\kcl - \mu_m^\kcl\right|^{2+\ell}\right) \\
    & = 2^{2+\ell} \underset{1\leq j \leq p}{\max} \; \nk^{-1} \sum_{m=1}^p \frac{\left| (\G)_{jm} \right|}{\sum_{l=1}^p \left| (\G)_{jl} \right|} \snk \Exp \left( \left| V_{im}^\kcl - \mu_m^\kcl\right|^{2+\ell}\right) \\
    & \leq 2^{2+\ell} \underset{1\leq m \leq p}{\max}\; \nk^{-1} \snk \Exp \left( \left| V_{im}^\kcl - \mu_m^\kcl\right|^{2+\ell}\right) \\
    & \leq 2^{2+\ell} B_{n}^\ell\\
    & \leq L_{n}^\ell,
\end{align*}
where the first inequality follows from the triangle inequality and the fact that by the definition of the matrix $\G,$ we have $\max_{1\leq j \leq p} \sum_{l=1}^p \left| (\G)_{jl} \right| \leq 2$. The second inequality follows from Jensen's inequality. The third inequality follows as we have taken the maximum over $m = 1,\dots,p,$ and the fourth inequality follows from Assumption (A.3). The final inequality follows by taking $L_{n} = 8 B_{n}.$

\noindent (c) We have
\begin{align*}
     & \underset{ 1\leq k\leq K }{\max}\; \underset{1\leq i \leq \nk}{\max}\; \underset{1\leq j \leq p}{\max} \; \Exp \left\{\exp\left( L_{n}^{-1}\left|Z_{ij}^\kcl - \nu_j^\kcl \right| \right) \right\}\\
     & = \underset{ 1\leq k\leq K }{\max}\; \underset{1\leq i \leq \nk}{\max}\; \underset{1\leq j \leq p}{\max} \; \Exp \left\{\exp\left( L_{n}^{-1} \left| \sum_{m=1}^p (\G)_{jm} \left(V_{im}^\kcl - \mu_m^\kcl\right) \right| \right) \right\}\\
     & \leq \underset{ 1\leq k\leq K }{\max}\; \underset{1\leq i \leq \nk}{\max}\; \underset{1\leq j \leq p}{\max} \; \Exp \left\{\exp\left( B_{n}^{-1} \sum_{m=1}^p \left(\frac{\left|(\G)_{jm} \right|}{\sum_{l=1}^p \left| (\G)_{jl} \right|}\right) \left|V_{im}^\kcl - \mu_m^\kcl \right| \right) \right\}\\
     & \leq \underset{ 1\leq k\leq K }{\max}\; \underset{1\leq i \leq \nk}{\max}\; \underset{1\leq j \leq p}{\max} \; \sum_{m=1}^p \left(\frac{\left|(\G)_{jm}\right|}{\sum_{l=1}^p \left| (\G)_{jl} \right|}\right) \Exp \left\{\exp\left( B_{n}^{-1} \left| V_{im}^\kcl - \mu_m^\kcl \right|\right) \right\} \\
     & \leq \underset{ 1\leq k\leq K }{\max}\; \underset{1\leq i \leq \nk}{\max}\; \underset{1\leq m \leq p}{\max} \; \Exp \left\{\exp\left( B_{n}^{-1} \left|  V_{im}^\kcl - \mu_m^\kcl \right| \right) \right\}  \leq 2,
\end{align*}
where the first inequality is due to the fact $\max_{1\leq j \leq p} \sum_{l=1}^p \left| (\G)_{jl} \right| \leq 2$ and $L_{n} = 8 B_{n}.$ The second inequality follows from Jensen's inequality. The third inequality follows by taking maximum over $m = 1,\dots,p,$ and the final inequality follows from Assumption (A.4).

\noindent (d) We have
\begin{align*}
     & \underset{ 1\leq k\leq K }{\max}\; \underset{1\leq i \leq \nk}{\max}\; \; \Exp \left\{\underset{1\leq j \leq p}{\max} \left( L_{n}^{-1} \left|Z_{ij}^\kcl - \nu_j^\kcl \right| \right)^q \right\}\\
     & = \underset{ 1\leq k\leq K }{\max}\; \underset{1\leq i \leq \nk}{\max}\; \Exp \left\{\underset{1\leq j \leq p}{\max} \left( \left| L_{n}^{-1} \sum_{m=1}^p (\G)_{jm} \left(V_{im}^\kcl - \mu_m^\kcl\right) \right| \right)^q \right\}\\
     & \leq \underset{ 1\leq k\leq K }{\max}\; \underset{1\leq i \leq \nk}{\max}\; \Exp \left\{\underset{1\leq j \leq p}{\max} \left( B_{n}^{-1} \sum_{m=1}^p \left(\frac{ \left|(\G)_{jm}\right| }{ \sum_{l=1}^p \left| (\G)_{jl} \right| } \right) \left|V_{im}^\kcl - \mu_m^\kcl \right|\right)^q \right\}\\
     & \leq \underset{ 1\leq k\leq K }{\max}\; \underset{1\leq i \leq \nk}{\max}\; \Exp \left\{\underset{1\leq m \leq p}{\max} \left( B_{n}^{-1} \left| V_{im}^\kcl - \mu_m^\kcl \right| \right)^q \right\}\\ 
     & \leq 2,
\end{align*}
where the first inequality is due to the fact $\max_{1\leq j \leq p} \sum_{l=1}^p \left| (\G)_{jl} \right| \leq 2$ and $L_{n} = 8 B_{n}.$ The second inequality follows by taking the maximum over $m = 1,\dots,p,$ and the final inequality follows from Assumption (A.5). \qed

\medskip

\noindent \textit{Proof of Lemma \ref{prop4}.}
For $\kK$ define $\bar\mu^\kcl = p^{-1} \sum_{j=1}^p \mu_j^\kcl.$ Then using triangle inequality we have
\begin{align*}
    \max_{\kKK} \| \nu^\kcl - \nu^\kclp \|_\infty
    & = \max_{\kKK} \left\|\G (\mu^\kcl - \mu^\kclp) \right\|_\infty\\
    & = \max_{\kKK} \left\|(\mu^\kcl - \mu^\kclp) - (\bar\mu^\kcl - \bar\mu^\kclp)1_p \right\|_\infty\\
    & \geq \max_{\kKK} \Big\{ \left\|\mu^\kcl - \mu^\kclp \right\|_\infty - \left|\bar\mu^\kcl - \bar\mu^\kclp \right| \Big\}\\
    & \geq \max_{\kKK} \Big\{\| \mu^\kcl - \mu^\kclp \|_\infty - p^{-1}\| \mu^\kcl - \mu^\kclp \|_1 \Big\}\\
    & \geq \Delta_{n,p}.
\end{align*} \qed

\medskip

\noindent \textit{Proof of Lemma \ref{lemma1}.}
Let ${Q}{P}^T = {M}$ and $M_{ij}$ be the $(i,j)$th element of ${M}.$ Then using triangle inequality and the condition of the lemma, we have
\begin{equation*}
    \|{Q}{P}^T\|_\infty = \underset{1\leq i,j \leq p}{\max} |M_{ij}|
    \leq \underset{1\leq i,j \leq p}{\max} \sum_{m=1}^p |Q_{im}| | P_{jm} |
    \leq \underset{1\leq i,k \leq p}{\max}|Q_{ik}| \; \underset{1\leq j \leq p}{\max} \sum_{m=1}^p | P_{jm} |
    \leq s \|{Q}\|_\infty.
\end{equation*}
Similar calculation shows $\|{P}{Q}\|_\infty \leq s \|{Q}\|_\infty.$
Combining these two, we get
\begin{equation*}
    \|{P}{Q}{P}^T\|_\infty \leq s \|{P}{Q}\|_\infty \leq s^2 \|{Q}\|_\infty.
\end{equation*} \qed

\medskip

\noindent \textit{Proof of Lemma \ref{lem-rho}.} (a) Recall the construction of the vector $U_t$ from the proof of Theroem 1. For $t =1,\dots,N,$ consider $U_t = {A} V_t \in \R^{K(K+1)p/2}$ as a collection of $K(K+1)/2$ sub-vectors of lengths $p$ each. For $\kK,$ and $i=1,\dots,n_k,$ set $t = i + N_{k-1},$ and we have $(K-k)$ sub-vectors of $U_t$ are equal to $\lamk N^{1/2}\nk^{-1/2} (Z_i^\kcl-\nu^\kcl)$ where $\kKK,$ $(k-1)$ sub-vectors of $U_t$ are equal to $-\lamk N^{1/2}\nk^{-1/2} (Z_i^\kcl-\nu^\kcl)$ where $1\leq l<k\leq K$ and the rests are all zeroes. Furthermore, for $\kKK$ and $j=1,\dots,p$ set $m=\Kkk p +j,$ where $K_{k,l} = \Big((k-1)K+l-k(k+1)/2-1\Big),$ and we have the $m$th entry of $N^{-1/2}\st U_t$ as $N^{-1/2}\st U_{tm} = \lamk S_{n_k j}^\kcl- \lamkk S_{n_{l} j}^\kclp .$

Then for $\kKK,$ and $m=(\Kkk p +1),\dots,(\Kkk+1)p,$ $\Exp(U_{tm}) = 0$ and
 \begin{equation} \label{eq1}
 \begin{split}
    N^{-1} \st \Exp \Big(U_{tm}^2\Big)
    &= \Exp \Big( \big({A}S_n\big)^2_m \Big)\\
    & = \Exp \Big(\lamk S_{n_k j}^\kcl- \lamkk S_{n_{l} j}^\kclp \Big)^2\\
    & = \lamk^2 \nk^{-1} \snk \Exp \Big(Z_{ij}^\kcl - \nu_{j}^\kcl\Big)^2  + \lamkk^2 \nkk^{-1} \snkk \Exp \Big(Z_{ij}^\kclp - \nu_{j}^\kclp\Big)^2 \\
    & \geq b_1,
    \end{split}
\end{equation}
where the final inequality follows from Assumption (A.2) and Lemma \ref{propm1}. Hence we get for all $m = 1,\dots,K(K+1)p/2,$ $N^{-1} \st \Exp \Big(U_{tm}^2\Big) \geq b_1.$ Note that, from the Assumption (A.1), for any $1\leq k\leq K,$ we get $N/n_k \leq K(c_2(1-c_2)^{-1}+(1-c_1)c_1^{-1}+1)$. From the Assumption (A.3) and Lemma \ref{propm1}, we get for $\ell=1,2,$
 \begin{align} \label{eq2}
    & \underset{1\leq m \leq K(K+1)p/2}{\max} N^{-1} \st \Exp \left(\left|U_{tm}\right|^{2+\ell}\right) \nonumber\\
    & = \underset{\kKK}{\max}\; \underset{(\Kkk-1)p+1\leq m \leq \Kkk p}{\max} N^{-1} \st \Exp \left(\left|U_{tm}\right|^{2+\ell}\right)\nonumber\\
    & = \underset{\kKK}{\max}\; \underset{1\leq j \leq p}{\max} N^{-1} \left[ \snk \Exp \left(\left| \lamk N^{1/2} \nk^{-1/2} \left(Z_{ij}^\kcl - \nu_j^\kcl\right)\right|^{2+\ell}\right) \right.\nonumber\\& \hspace{4cm} \left. + \snkk \Exp \left(\left| \lamkk N^{1/2} \nkk^{-1/2} \left(Z_{ij}^\kclp - \nu_j^\kclp\right)\right|^{2+\ell}\right) \right] \displaybreak[1] \nonumber\\
    & \leq \underset{\kKK}{\max}\; \underset{1\leq j \leq p}{\max} (N/n_k)^{\ell/2} \lamk^{2+\ell} \nk^{-1} \snk \Exp \left(\left|Z_{ij}^\kcl - \nu_j^\kcl\right|^{2+\ell}\right) \nonumber\\&\hspace{4cm} + \underset{\kKK}{\max}\; \underset{1\leq j \leq p}{\max} (N/n_{l})^{\ell/2} \lamkk^{2+\ell} \nkk^{-1} \snkk \Exp \left(\left|Z_{ij}^\kclp - \nu_j^\kclp\right|^{2+\ell}\right)\nonumber\\
    & \leq C L_{n}^\ell \leq G_{n}^\ell,
\end{align}
where $G_n$ is a sequence that depends on $B_n$. If (A.4) holds, then by Lemma \ref{propm1} we have
\begin{align*}
     & \underset{1\leq t \leq n}{\max}\; \underset{1\leq m \leq K(K+1)p/2}{\max} \; \Exp \left\{\exp( G_{n}^{-1} \left|U_{tm}\right| ) \right\}\\
     & \leq \underset{ 1\leq k\leq K }{\max}\; \underset{1\leq i \leq n}{\max}\; \underset{1\leq j \leq p}{\max} \; \left[\Exp \left\{\exp( L_{n}^{-1} \sqrt{n/(Cn_k)} \left|Z_{ij}^\kcl - \nu_j^\kcl \right| ) \right\} \vee 1\right]\\
     & \leq \underset{ 1\leq k\leq K }{\max}\; \underset{1\leq t \leq n}{\max}\; \underset{1\leq j \leq p}{\max} \; \left[\Exp \left\{\exp( L_{n}^{-1} \left|Z_{ij}^\kcl - \nu_j^\kcl \right|) \right\} \vee 1\right]\\
     & \leq 2,
\end{align*}
where the first inequality follows from the fact that for any $t=1,\dots, N$ and $m=1,\dots, K(K+1)p/2,$ $U_{tm}$ is either $0$ or $\lambda_{k,l} N^{1/2} \nk^{-1/2} \left|Z_{ij}^\kcl - \nu_j^\kcl \right|$ for some $k\neq l = 1,\dots, K,\;i=1,\dots,n_k,\;\text{ and }j=1,\dots,p$ based on $t$ and $m$ and $\lambda_{k,l} \leq 1$ for any $k\neq l = 1,\dots, K.$ The second inequality follows from the fact that for any $k=1,\dots, K,$ $N/n_k\leq C.$


Therefore, by Proposition 2.1 in Chernozhukov, Chetverikov, and Kato (2017), we get
\begin{equation*}
    \rho({\tilde T_{n,K}, T_{n,K}^{(N)}}) \leq C(b_1,c_1,c_2) \Big(G^2_{n} \log^7(Np)/N\Big)^{1/6} \leq C \Big(B^2_{N} \log^7(Np)/N\Big)^{1/6} .
\end{equation*}

To bound the quantity $\rho^*(T_{n,K}^{(N)}, T_{n,K}^{(G)}),$ note that we can rewrite $\hat\Gamma_u$ as
\begin{equation*}
    \hat\Gamma_u = \text{Cov}_*({A}S_n^*) = {A} \text{Cov}_*(S_n^*){A}^T = {A}\hat\Gamma {A}^T,
\end{equation*}
where $\text{Cov}_*(.)$ is the covariance matrix computed conditional on the data and $\hat\Gamma = \text{Cov}_*(S_n^*) = diag(\hat\Gamma^\kone,\dots,\hat\Gamma^\kk),$ and $\hat\Gamma^\kcl = \text{Cov}_*(Z_1^{*\kcl}) =  \nk^{-1}\snk (Z_i^\kcl - \bar{Z}^\kcl) $ $(Z_i^\kcl - \bar{Z}^\kcl)^T,$ for $\kK$. Similarly, we can rewrite $\Gamma_u$ as
 \begin{align*}
     \Gamma_u = N^{-1}\st \Exp(U_tU_t^T) = N^{-1}\st {A}\Exp(V_tV_t^T) {A}^T = {A}\Gamma {A}^T,
 \end{align*}
 where $\Gamma = N^{-1}\st \Exp(V_tV_t^T) = Diag(\Gamma^\kone,\dots,\Gamma^\kk)$ and $\Gamma^\kcl = \G \Sigma^\kcl \G$ for $\kK$.
For $\kK$, define 
\begin{align*}
    &\hat\Delta_{n}^{(2k-1)} = \underset{1\leq j,k\leq p}{\max} \Big|\nk^{-1} \snk \Big\{ (Z_{ij}^\kcl - \nu_j^\kcl)(Z_{ik}^\kcl - \nu_k^\kcl) - \Exp(Z_{ij}^\kcl - \nu_j^\kcl)(Z_{ik}^\kcl - \nu_k^\kcl) \Big\} \Big|,\\
    & \hat\Delta_{n}^{(2k)} = \underset{1\leq j\leq p}{\max} \Big| \bar{Z}_{j}^\kcl - \nu_j^\kcl \Big|.
    \end{align*}
Then, for any $\kK,$ we have
\begin{equation*}
    \|\hat\Gamma^\kcl - \Gamma^\kcl\|_{\infty} \leq \hat\Delta_{n}^{(2k-1)} + {(\hat\Delta_{n}^{(2k)})}^2.
\end{equation*}
Define,
    \begin{align*}
    \hat\Delta_{n} = \max_{1\leq k\leq K} \{\hat\Delta_{n}^\kcl + (\hat\Delta_{n}^{(k+2)})^2\}.
\end{align*}
Note that, 
\begin{equation*}
 \underset{1\leq m\leq K(K+1)p/2}{\max} \sum_{k=1}^{Kp}|A_{mk}| = \underset{\kKK}{\max}  (\lamk + \lamkk)  \leq 2.
\end{equation*}
Then, by Lemma \ref{lemma1}, it follows that
\begin{align} \label{m13}
    \| \hat\Gamma_u - \Gamma_u \|_\infty \leq \|{A}(\hat\Gamma - \Gamma){A}^T\|_{\infty} \leq 4\|\hat\Gamma - \Gamma\|_{\infty}\leq 4\max_{1\leq k\leq K}\|\hat\Gamma^\kcl - \Gamma^\kcl\|_{\infty}\leq C \hat\Delta_{n}.
\end{align}
Let, $\bar\Delta_{n}$ be a positive real number. By Proposition 2.1 in \cite{chernozhukov2022improved}, on the event $E=\{ \hat\Delta_{n} \leq \bar\Delta_{n} \}$, we get
\begin{align*}
    \rho^*({T_{n,K}^{(N)}, T_{n,K}^{(G)}}) \leq C(b_1,c_1,c_2) {\bar\Delta_{n}}^{1/2} \log(p),
\end{align*}

If Assumption (A.4) holds, then we take $\bar\Delta_{n} = \Big(G^2_{n} \log(Np) \log^2(N) /N\Big)^{1/2}.$ By Proposition 4.1 in \cite{chernozhukov2017clt}, we have $\Prob(E) \geq 1-CN^{-1}.$ Therefore with probability at least $1-CN^{-1},$ we get
\begin{equation} \label{m15}
    \begin{split}
         \rho^*({T_{n,K}^{(N)}, T^{(G)}_{n,K}}) &\leq C(b_1,c_1,c_2) \Big(G^2_{n} \log^5(Np) \log^2(N)/N\Big)^{1/4} \\&\leq C \Big(B^2_{n} \log^5(Np) \log^2(N)/N\Big)^{1/4}.
    \end{split}
\end{equation} 
Next, we find a bound for the quantity $\rho^{**}({ T^*_{n,K}, T_{n,K}^{(G)}}).$ For $\phi\geq 1,$ define
\begin{equation*}
    \begin{split}
        &\widehat L_n = \max_{1\leq m\leq K(K+1)p/2} N^{-1} \st \Exp_*|U_{tm}^*|^3,\\
        &\widehat{M}_{n,U}(\phi) = N^{-1} \st \Exp_*\bigg[ \max_{1\leq m\leq K(K+1)p/2} |U_{tm}^*|^3 I \bigg\{ \max_{1\leq m\leq K(K+1)p/2} |U_{tm}^*| > n^{1/2} (4\phi\log(p))^{-1} \bigg\} \bigg],\\
        &\widehat{M}_{n,G}(\phi) = \Exp_*\bigg[ \max_{1\leq m\leq K(K+1)p/2} |S_{n,K}^{(G)}|^3 I \bigg\{ \max_{1\leq m\leq K(K+1)p/2} |S_{n,K}^{(G)}| > n^{1/2} (4\phi\log(p))^{-1} \bigg\} \bigg],\\
        &\widehat M_n(\phi) = \widehat{M}_{n,U}(\phi) + \widehat{M}_{n,G}(\phi).
    \end{split}
\end{equation*}
Then, applying Theorem 2.1 of \cite{chernozhukov2017clt}, conditional on the data, we can find constants $\widehat L_n$ and $\widehat M_n$, such that
\begin{equation} \label{rhoeb}
    \rho^{**}({ T^*_{n,K}, T_{n,K}^{(G)}}) \leq C(b_1) \Big\{ \Big(\widehat L_n^2 \log^7(p) N^{-1}\Big)^{1/6}  + \widehat M_n \widehat L_n^{-1} \Big\},
\end{equation}
on the set $F = \Big\{N^{-1} \st (U_{tm} - \bar U_{Nm})^2 \geq b_1, \text{ for all } m=1,\dots,K(K+1)p/2\Big\} \cap \Big\{ \widehat L_n \leq  \overline L_n \Big\} \cap \Big\{ \widehat M_n\Big(C(b_1) (\overline L_n^2 \log^4(p) N^{-1})^{-1/6}\Big) \leq  \overline M_n \Big\}.$

From \eqref{m13}, it follows that for large enough $N,$ $\Prob(\| \hat\Gamma_u - \Gamma_u \|_\infty \leq b_1/2)\geq 1 - CN^{-1}.$ Furthermore, following similar calculations in \eqref{eq2} with $\ell = 0,$ it is straightforward to show that $N^{-1}\st \Exp (U_{tm}^2) \leq G_n \leq CB_n,$ for all $ m=1,\dots, K(K+1)p/2$. Consequently from \eqref{eq1} and \eqref{m13}, it follows that $ b_1/2 \leq N^{-1} \st (U_{tm} - \bar U_{Nm})^2 \leq C B_n$ with probability at least $1 - CN^{-1}.$

Therefore to quantify \eqref{rhoeb}, it is enough to find bounds for the quantities $\widehat L_n$ and $\widehat M_n$ or find a value for $\overline L_n$ and $\overline M_n$ such that the probability of the set $F$ tends to one. To find bounds for $\widehat L_n$ and $\widehat M_n$, we define
\begin{equation*}
    \begin{split}
        \widehat L_{n}^\kcl &= \max_{1\leq j\leq p} \nk^{-1} \snk |Z_{ij}^\kcl - \bar Z^\kcl|^3,\\
        \widehat M_{n,Z}^\kcl(\phi) &= \max_{1\leq j\leq p} \nk^{-1} \snk |Z_{ij}^\kcl - \bar Z^\kcl|^3  I \bigg\{ \max_{1\leq j\leq p}  \snk |Z_{ij}^\kcl - \bar Z^\kcl| > n^{1/2} (4\phi \log(p))^{-1} \bigg\},
    \end{split}
\end{equation*}
for $\kK$ and $\phi\geq 1.$
Following the similar calculations as in \eqref{eq2} and the fact that for $\kK,$ $Z_1^{*\kcl},\dots,Z_\nk^{*\kcl}$ are independent and identical samples from the empirical distribution of $Z^\kcl$, we get
\begin{equation*}
    \begin{split}
        \widehat L_n &= \max_{1\leq m\leq K(K+1)p/2} N^{-1} \st \Exp_*|U_{tm}^*|^3,\\
        &= \max_{\kKK} \max_{1\leq j\leq p} N^{1/2} \bigg\{ \nk^{-3/2} \lamk^3  \snk |Z_{ij}^\kcl - \bar Z^\kcl|^3 + \nkk^{-3/2} \lamkk^3  \snkk |Z_{ij}^\kclp - \bar Z^\kclp|^3\bigg\}\\
        &\leq 2C^{1/2} \max_{1\leq k\leq K} \max_{1\leq j\leq p} \nk^{-1} \snk |Z_{ij}^\kcl - \bar Z^\kcl|^3\\
        &= 2C^{1/2} \max_{1\leq k\leq K} \widehat L_n^\kcl.
    \end{split}
\end{equation*}
Following similar arguments, we get
\begin{align*}
        \widehat{M}_{n,U}(\phi) &= N^{-1} \st \Exp_*\bigg[ \max_{1\leq m\leq K(K+1)p/2} |U_{tm}^*|^3 I \bigg\{ \max_{1\leq m\leq K(K+1)p/2} |U_{tm}^*| > n^{1/2} (4\phi\log(p))^{-1} \bigg\} \bigg] \displaybreak[1]\\
        &= N^{-1} \sK \sum_{t=N_{k-1}+1}^{N_k} \Exp_*\bigg[ \max_{1\leq m\leq K(K+1)p/2} |U_{tm}^*|^3 I \bigg\{ \max_{1\leq m\leq K(K+1)p/2} |U_{tm}^*| > n^{1/2} (4\phi\log(p))^{-1} \bigg\} \bigg]\\
        &= N^{1/2} \sK \snk \underset{l\neq k}{\max_{1\leq l \leq K}} \max_{1\leq j\leq p} \nk^{-3/2} \lamk^3  \snk |Z_{ij}^\kcl - \bar Z^\kcl|^3 \\&\hspace{3cm} I \bigg\{ \underset{l\neq k}{\max_{1\leq l \leq K}} \max_{1\leq j\leq p} N\nk^{-1} \lamk  \snk |Z_{ij}^\kcl - \bar Z^\kcl| > n^{1/2} (4\phi\log(p))^{-1} \bigg\}\\
        &\leq K C^{1/2} c_4^3 \max_{1\leq k\leq K} \max_{1\leq j\leq p} \nk^{-1} \snk |Z_{ij}^\kcl - \bar Z^\kcl|^3 I \bigg\{ \max_{1\leq j\leq p}  \snk |Z_{ij}^\kcl - \bar Z^\kcl| > n^{1/2} (4\phi'\log(p))^{-1} \bigg\}\\
        &= K C^{1/2} c_4^3 \max_{1\leq k\leq K} \widehat{M}^\kcl_{n,Z}(\phi),
\end{align*}
where $c_4 = \max\{ c_2^{1/2}, (1-c_1)^{1/2} \}$ and $\phi' = Cc_3\phi$. Following the proof of Proposition 4.3 of \cite{chernozhukov2017clt}, for $\kK,$ we can obtain
\begin{equation*}
\begin{split}
    &\Prob\Big(\widehat L_n^\kcl > CB_n \Big) \leq CN^{-1}, \\
    &\Prob\Bigg(\widehat{M}^\kcl_{n,Z}\Big(C (\overline L_n^2 \log^4(p) N^{-1})^{-1/6}\Big) > 0 \Bigg) \leq CN^{-1},
\end{split}
\end{equation*}
and
\begin{equation*}
    \Prob\Bigg(\widehat{M}_{n,G}\Big(C (\overline L_n^2 \log^4(p) N^{-1})^{-1/6}\Big) > 0 \Bigg) \leq CN^{-1}.
\end{equation*}
Therefore, using union bound, it follows that 
\begin{equation*}
\begin{split}
    &\Prob\Big(\widehat L_n \leq CB_n \Big) \geq 1-CN^{-1}, \\
    &\Prob\Bigg(\widehat{M}_{n,U}\Big(C (\overline L_n^2 \log^4(p) N^{-1})^{-1/6}\Big) = 0 \Bigg) \geq 1-CN^{-1}.
\end{split}
\end{equation*}
Finally, from \eqref{rhoeb}, with probability at least $1-CN^{-1},$ we get
\begin{equation*}
    \rho^{**}( T_{n,K}^*, T_{n,K}^{(G)}) \leq  C \Big(B^2_{n} \log^7(Np)/N\Big)^{1/6}.
\end{equation*}

\noindent (b) The proof of this part follows by using Assumption (A.5) instead of Assumption (A.4) in the calculations of part a of this proof.
\qed

\subsection*{Appendix C: additional simulation results}\label{sec: sim}

\begin{table}[ht]
\centering
\small
\caption{Empirical size and power under Setting (M.3) with CZM imputation and additional zero inflation level $\eta = 0.45$. The average proportion of zeros before imputation is around $55\%$--$65\%$.}
\label{Tables1}
\begin{adjustbox}{max width=0.85\textwidth}
\begin{threeparttable}
\begin{tabular}{cclcccccccccc}
\toprule
 & & & \multicolumn{10}{c}{$\delta$} \\
\cmidrule(lr){4-13}
$p$ & Dist. & Method & \multicolumn{2}{c}{0} & \multicolumn{2}{c}{0.25} & \multicolumn{2}{c}{0.50} & \multicolumn{2}{c}{0.75} & \multicolumn{2}{c}{1} \\
\cmidrule(lr){4-5}\cmidrule(lr){6-7}\cmidrule(lr){8-9}\cmidrule(lr){10-11}\cmidrule(lr){12-13}
 & & & $N_1$ & $N_2$ & $N_1$ & $N_2$ & $N_1$ & $N_2$ & $N_1$ & $N_2$ & $N_1$ & $N_2$ \\
\midrule
200 & D1 & EBC & 0.045 & 0.059 & 0.068 & 0.076 & 0.095 & 0.165 & 0.190 & 0.471 & 0.366 & 0.804 \\
 &  & CLL & 0.084 & 0.062 & 0.092 & 0.083 & 0.118 & 0.173 & 0.179 & 0.418 & 0.366 & 0.755 \\
 &  & MECAF & 0.074 & 0.060 & 0.086 & 0.079 & 0.101 & 0.167 & 0.162 & 0.404 & 0.348 & 0.741 \\
\addlinespace[1pt]
 & D2 & EBC & 0.049 & 0.064 & 0.057 & 0.075 & 0.084 & 0.195 & 0.168 & 0.480 & 0.403 & 0.843 \\
 &  & CLL & 0.081 & 0.077 & 0.081 & 0.088 & 0.108 & 0.197 & 0.177 & 0.453 & 0.382 & 0.793 \\
 &  & MECAF & 0.074 & 0.072 & 0.075 & 0.084 & 0.098 & 0.190 & 0.162 & 0.444 & 0.363 & 0.786 \\
\midrule
500 & D1 & EBC & 0.047 & 0.088 & 0.049 & 0.111 & 0.094 & 0.225 & 0.248 & 0.667 & 0.629 & 0.970 \\
 &  & CLL & 0.076 & 0.097 & 0.082 & 0.122 & 0.116 & 0.209 & 0.233 & 0.561 & 0.460 & 0.934 \\
 &  & MECAF & 0.070 & 0.092 & 0.075 & 0.115 & 0.104 & 0.202 & 0.210 & 0.547 & 0.431 & 0.928 \\
\addlinespace[1pt]
 & D2 & EBC & 0.036 & 0.090 & 0.054 & 0.101 & 0.092 & 0.285 & 0.286 & 0.715 & 0.608 & 0.980 \\
 &  & CLL & 0.081 & 0.095 & 0.093 & 0.109 & 0.134 & 0.247 & 0.232 & 0.592 & 0.472 & 0.943 \\
 &  & MECAF & 0.073 & 0.090 & 0.082 & 0.101 & 0.123 & 0.233 & 0.204 & 0.581 & 0.448 & 0.941 \\
\midrule
1000 & D1 & EBC & 0.041 & 0.047 & 0.041 & 0.077 & 0.089 & 0.272 & 0.322 & 0.809 & 0.743 & 0.996 \\
 &  & CLL & 0.085 & 0.056 & 0.103 & 0.094 & 0.124 & 0.225 & 0.251 & 0.659 & 0.531 & 0.980 \\
 &  & MECAF & 0.077 & 0.050 & 0.085 & 0.087 & 0.105 & 0.209 & 0.226 & 0.639 & 0.494 & 0.976 \\
\addlinespace[1pt]
 & D2 & EBC & 0.045 & 0.057 & 0.050 & 0.064 & 0.104 & 0.258 & 0.330 & 0.827 & 0.754 & 1.000 \\
 &  & CLL & 0.074 & 0.074 & 0.100 & 0.069 & 0.130 & 0.224 & 0.248 & 0.703 & 0.525 & 0.979 \\
 &  & MECAF & 0.067 & 0.068 & 0.081 & 0.065 & 0.113 & 0.208 & 0.218 & 0.685 & 0.493 & 0.976 \\
\bottomrule
\end{tabular}
\end{threeparttable}
\end{adjustbox}
\end{table}

\begin{table}[ht]
\centering
\caption{Empirical size and power under Setting (M.3) with CZM imputation and additional zero inflation level $\eta = 0.60$. The average proportion of zeros before imputation is around $68\%$--$75\%$.}
\label{Tables2}
\begin{adjustbox}{max width=0.85\textwidth}
\begin{threeparttable}
\begin{tabular}{cclcccccccccc}
\toprule
 & & & \multicolumn{10}{c}{$\delta$} \\
\cmidrule(lr){4-13}
$p$ & Dist. & Method & \multicolumn{2}{c}{0} & \multicolumn{2}{c}{0.25} & \multicolumn{2}{c}{0.50} & \multicolumn{2}{c}{0.75} & \multicolumn{2}{c}{1} \\
\cmidrule(lr){4-5}\cmidrule(lr){6-7}\cmidrule(lr){8-9}\cmidrule(lr){10-11}\cmidrule(lr){12-13}
 & & & $N_1$ & $N_2$ & $N_1$ & $N_2$ & $N_1$ & $N_2$ & $N_1$ & $N_2$ & $N_1$ & $N_2$ \\
\midrule
200 & D1 & EBC & 0.064 & 0.085 & 0.060 & 0.091 & 0.065 & 0.145 & 0.119 & 0.262 & 0.226 & 0.536 \\
 &  & CLL & 0.088 & 0.091 & 0.088 & 0.098 & 0.095 & 0.140 & 0.135 & 0.245 & 0.199 & 0.451 \\
 &  & MECAF & 0.078 & 0.086 & 0.076 & 0.092 & 0.083 & 0.136 & 0.118 & 0.235 & 0.188 & 0.435 \\
\addlinespace[1pt]
 & D2 & EBC & 0.059 & 0.089 & 0.052 & 0.091 & 0.075 & 0.131 & 0.119 & 0.292 & 0.211 & 0.577 \\
 &  & CLL & 0.077 & 0.096 & 0.082 & 0.094 & 0.095 & 0.142 & 0.137 & 0.276 & 0.211 & 0.508 \\
 &  & MECAF & 0.070 & 0.093 & 0.071 & 0.090 & 0.093 & 0.133 & 0.125 & 0.266 & 0.197 & 0.494 \\
\midrule
500 & D1 & EBC & 0.048 & 0.109 & 0.031 & 0.111 & 0.062 & 0.201 & 0.171 & 0.491 & 0.346 & 0.868 \\
 &  & CLL & 0.076 & 0.129 & 0.062 & 0.141 & 0.083 & 0.201 & 0.126 & 0.389 & 0.181 & 0.656 \\
 &  & MECAF & 0.065 & 0.126 & 0.054 & 0.132 & 0.074 & 0.194 & 0.104 & 0.376 & 0.163 & 0.640 \\
\addlinespace[1pt]
 & D2 & EBC & 0.034 & 0.119 & 0.042 & 0.117 & 0.054 & 0.249 & 0.146 & 0.507 & 0.378 & 0.862 \\
 &  & CLL & 0.064 & 0.134 & 0.071 & 0.151 & 0.101 & 0.224 & 0.122 & 0.379 & 0.200 & 0.644 \\
 &  & MECAF & 0.053 & 0.129 & 0.065 & 0.141 & 0.093 & 0.212 & 0.111 & 0.366 & 0.185 & 0.628 \\
\midrule
1000 & D1 & EBC & 0.024 & 0.035 & 0.034 & 0.047 & 0.041 & 0.134 & 0.152 & 0.526 & 0.479 & 0.930 \\
 &  & CLL & 0.053 & 0.043 & 0.053 & 0.058 & 0.066 & 0.104 & 0.088 & 0.291 & 0.191 & 0.696 \\
 &  & MECAF & 0.046 & 0.038 & 0.048 & 0.054 & 0.054 & 0.101 & 0.076 & 0.281 & 0.175 & 0.678 \\
\addlinespace[1pt]
 & D2 & EBC & 0.024 & 0.049 & 0.028 & 0.054 & 0.056 & 0.146 & 0.184 & 0.546 & 0.466 & 0.944 \\
 &  & CLL & 0.043 & 0.049 & 0.051 & 0.063 & 0.072 & 0.116 & 0.107 & 0.307 & 0.218 & 0.693 \\
 &  & MECAF & 0.037 & 0.047 & 0.041 & 0.059 & 0.065 & 0.110 & 0.100 & 0.289 & 0.190 & 0.675 \\
\bottomrule
\end{tabular}
\end{threeparttable}
\end{adjustbox}
\end{table}

\newpage
\begin{table}[ht]
\centering
\small
\caption{Empirical size and power for the multi-sample EBC test with CZM imputation and additional zero inflation level $\eta = 0.60$. The average proportion of zeros before imputation is around $65\%$--$70\%$ percent.}
\label{Table7}
\begin{adjustbox}{max width=0.85\textwidth}
\begin{threeparttable}
\begin{tabular}{ccl*{10}{c}}
\toprule
 & & & \multicolumn{10}{c}{$\delta$} \\
\cmidrule(lr){4-13}
$p$ & Dist. & $K$
& \multicolumn{2}{c}{0}
& \multicolumn{2}{c}{0.5}
& \multicolumn{2}{c}{1}
& \multicolumn{2}{c}{1.5}
& \multicolumn{2}{c}{2} \\
\cmidrule(lr){4-5}\cmidrule(lr){6-7}\cmidrule(lr){8-9}\cmidrule(lr){10-11}\cmidrule(lr){12-13}
 & &
& $N_3$ & $N_4$
& $N_3$ & $N_4$
& $N_3$ & $N_4$
& $N_3$ & $N_4$
& $N_3$ & $N_4$ \\
\midrule
200 & D1 & $K=3$ & 0.036 & 0.039 & 0.052 & 0.037 & 0.061 & 0.066 & 0.098 & 0.234 & 0.197 & 0.540 \\
 &  & $K=4$ & 0.037 & 0.041 & 0.046 & 0.053 & 0.112 & 0.267 & 0.347 & 0.811 & 0.695 & 0.988 \\
 &  & $K=5$ & 0.029 & 0.041 & 0.057 & 0.104 & 0.287 & 0.773 & 0.805 & 1.000 & 0.984 & 1.000 \\
\addlinespace[1pt]
200 & D2 & $K=3$ & 0.038 & 0.038 & 0.046 & 0.051 & 0.060 & 0.071 & 0.096 & 0.243 & 0.182 & 0.543 \\
 &  & $K=4$ & 0.029 & 0.045 & 0.048 & 0.063 & 0.107 & 0.286 & 0.344 & 0.810 & 0.715 & 0.994 \\
 &  & $K=5$ & 0.032 & 0.046 & 0.066 & 0.110 & 0.306 & 0.763 & 0.802 & 0.997 & 0.987 & 1.000 \\
\midrule
500 & D1 & $K=3$ & 0.050 & 0.048 & 0.046 & 0.050 & 0.047 & 0.095 & 0.125 & 0.368 & 0.313 & 0.792 \\
 &  & $K=4$ & 0.057 & 0.042 & 0.054 & 0.067 & 0.132 & 0.477 & 0.563 & 0.977 & 0.935 & 1.000 \\
 &  & $K=5$ & 0.044 & 0.037 & 0.060 & 0.135 & 0.472 & 0.957 & 0.983 & 1.000 & 1.000 & 1.000 \\
\addlinespace[1pt]
500 & D2 & $K=3$ & 0.034 & 0.042 & 0.031 & 0.051 & 0.048 & 0.106 & 0.118 & 0.374 & 0.323 & 0.799 \\
 &  & $K=4$ & 0.035 & 0.036 & 0.037 & 0.065 & 0.143 & 0.508 & 0.557 & 0.978 & 0.944 & 1.000 \\
 &  & $K=5$ & 0.030 & 0.035 & 0.055 & 0.152 & 0.498 & 0.969 & 0.976 & 1.000 & 1.000 & 1.000 \\
\midrule
1000 & D1 & $K=3$ & 0.015 & 0.037 & 0.023 & 0.035 & 0.063 & 0.122 & 0.150 & 0.522 & 0.443 & 0.935 \\
 &  & $K=4$ & 0.019 & 0.032 & 0.036 & 0.061 & 0.180 & 0.634 & 0.734 & 0.997 & 0.991 & 1.000 \\
 &  & $K=5$ & 0.029 & 0.031 & 0.061 & 0.176 & 0.652 & 0.996 & 0.999 & 1.000 & 1.000 & 1.000 \\
\addlinespace[1pt]
1000 & D2 & $K=3$ & 0.027 & 0.044 & 0.039 & 0.054 & 0.046 & 0.136 & 0.131 & 0.508 & 0.459 & 0.942 \\
 &  & $K=4$ & 0.028 & 0.038 & 0.048 & 0.076 & 0.187 & 0.621 & 0.741 & 0.998 & 0.995 & 1.000 \\
 &  & $K=5$ & 0.027 & 0.044 & 0.060 & 0.185 & 0.652 & 1.000 & 0.999 & 1.000 & 1.000 & 1.000 \\
\bottomrule
\end{tabular}
\end{threeparttable}
\end{adjustbox}
\end{table}

\clearpage
\section*{Acknowledgements}
This research is supported by funding from NIH UL1 TR002345 and R01 HD095986.
\bibliographystyle{plainnat}
\bibliography{reference}

\end{document}